\documentclass[12pt, journal, draftclsnofoot, onecolumn]{IEEEtran}

\IEEEoverridecommandlockouts

\usepackage{color}
\usepackage[dvipsnames]{xcolor}
\usepackage{balance}

\usepackage{amsthm}
\usepackage{mathtools}

\usepackage{bbm}
\usepackage{amssymb}
\usepackage{tikz}

\usepackage{enumitem}
\usepackage[algo2e,ruled,vlined,linesnumbered]{algorithm2e}
\SetInd{0em}{0.5em}
\SetKwComment{Comment}{/\!\!/}{}
\usepackage{booktabs}
\usepackage{cite}
\usepackage{caption}
\usepackage{subfig}

\newcommand\blfootnote[1]{%
  \begingroup
  \renewcommand\thefootnote{}\footnote{#1}%
  \addtocounter{footnote}{-1}%
  \endgroup
}

\newtheorem{Theorem}{Theorem}
\newtheorem{Proposition}{Proposition}

\newtheorem{Definition}{Definition}
\newtheorem{Construction}{Construction}

\newtheorem{Example}{Example}

\usepackage[normalem]{ulem} 

\newcommand{\citep}[1]{\cite{#1}}

\renewcommand{\tilde}{\widetilde}
\renewcommand{\hat}{\widehat}

\def\beq{\begin{equation}}
\def\eeq{\end{equation}}
\def\beqa{\begin{eqnarray}}
\def\eeqa{\end{eqnarray}}
\def\beqan{\begin{eqnarray*}}
\def\eeqan{\end{eqnarray*}}

\DeclareMathOperator*{\argmin}{arg\,min}
\DeclareMathOperator*{\argmax}{arg\,max}

\newcounter{newenumi}
\def\Lset{{\cal L}}
\def\Uset{{\cal U}}
\def\Iset{{\cal I}}
\def\Sset{{\cal S}}

\def\Pset{{\cal P}}

\def\MCLset{{\cal MCL}}

\def\Pset{{\mathcal P}}

\newcommand{\Beam}{\text{\sf Beam}}

\newcommand{\str}{^{*}}

\begin{document}

\title{Optimal Single-User Interactive Beam Alignment with Feedback Delay}
\author{
\IEEEauthorblockN{
Abbas Khalili, 
Shahram Shahsavari, 
Mohammad A. (Amir) Khojastepour, 
and Elza Erkip
} 
}

\maketitle
\begin{abstract}

Communication in Millimeter wave (mmWave) band relies on narrow beams due to directionality, high path loss, and shadowing. One can use beam alignment (BA) techniques to find and adjust the direction of these narrow beams. In this paper, BA at the base station (BS) is considered, where the BS sends a set of BA packets to scan different angular regions while the user listens to the channel and sends feedback to the BS for each received packet. It is assumed that the packets and feedback received at the user and BS, respectively, can be correctly decoded. Motivated by practical constraints such as propagation delay, a feedback delay for each BA packet is considered. At the end of the BA, the BS allocates a narrow beam to the user including its angle of departure for data transmission and the objective is to maximize the resulting expected beamforming gain. A general framework for studying this problem is proposed based on which a lower bound on the optimal performance as well as an optimality achieving scheme are obtained. Simulation results reveal significant performance improvements over the state-of-the-art BA methods in the presence of feedback delay.

\blfootnote{This work is supported by
National Science Foundation grants 1547332, 
1824434, and NYU WIRELESS Industrial Affiliates. A part of this work has been presented in ISIT 2021 \cite{fullVersion}.}
\end{abstract}


\section{Introduction}
Millimeter wave (mmWave) communications have the potential to significantly improve the data rate and latency of wireless networks \cite{mmWave-survey-nyu}. However, signals transmitted in mmWave frequency bands (between $30$ GHz and $300$ Hz) suffer from high path-loss and intense shadowing \cite{rappaport2011state}. As a solution, beamforming has been proposed to improve the signal strength by employing directional beams (i.e., beams with a narrow beam width) to concentrate the transmit power towards the directions of interest \cite{kutty2016beamforming}. 

Several experimental studies such as \cite{akdeniz2014millimeter} have shown that the mmWave channels only incorporate a few spatial clusters. Consequently, a procedure called beam alignment (BA) (a.k.a. beam training and beam search) is required to find narrow beams aligned with the angle of departure (AoD) of channel spatial clusters at the transmitter \cite{giordani2018tutorial}. In BA, the wireless transmitter uses beamforming to scan the space and match its transmission beam to the channel clusters. In general, BA can be used for initial access when the users try to connect to the base station (BS) and access network resources \cite{barati2016initial,giordani2016comparative} or for beam tracking where the beam directions is updated due to mobility or change in the propagation environment \cite{Scalabrin2018,Shah1906:Robust}. To reduce power consumption, it is suggested to use analog beamforming when performing BA in which the antenna array is connected to a single RF chain and hence a single direction at any given time is scanned. 

In general, BA schemes can be classified into two main categories, namely  \textit{interactive} and \textit{non-interactive} BA. To elaborate, let us consider the BA procedure in a downlink scenario at the BS whose objective is to localize the AoD of a user's channel. During BA, the BS sends a set of probing packets through a set of scanning beams in order to scan its angular space and after receiving user's feedback to all of the packets, decides on a narrow beam including the user AoD. In non-interactive BA, the BS uses a set of predetermined scanning beams which are independent of the feedback received during BA. In interactive BA, however, the BS uses the received feedback during BA to refine the scanning beams and better localize the user AoD compared to non-interactive BA.

We studied optimal non-interactive BA in \cite{khalili2020optimal}, where we investigated BA through an information theoretic perspective and provided bounds on the minimum expected data beamwidth along with optimality achievablity BA schemes. In this paper, we focus on interactive BA, which is more challenging due to the presence of the receiver's feedback during the BA. There is a large literature on interactive BA methods \cite{hussain2017throughput,michelusi2018optimal,chiu2019active,khalili2021single,yildiz2021hybrid,desai2014initial,khosravi2019efficient,shabara2018linear,klautau20185g,Song2019,7947209}. 
More specifically, \cite{khalili2021single,hussain2017throughput,michelusi2018optimal,chiu2019active} consider the problem of interactive analog BA for the single-path channel. References \cite{hussain2017throughput,michelusi2018optimal}, show that bisection search is optimal under the objectives of throughput maximization and average rate maximization subject to a average power consumption, respectively. In \cite{chiu2019active}, the authors provide a  active learning based BA strategy for finding a data beam with a target resolution. In \cite{khalili2021single}, we provide a deep learning based BA formulation using recurrent neural networks to minimize the expected user's data beamwidth. The impact of having multiple RF-chains is investigated in \cite{yildiz2021hybrid}, where we consider interactive hybrid BA by providing a group testing framework and develop novel group testing based BA strategies.

Prior works on interactive BA assume that the receiver's feedback on a probing packet is instantaneous and available before the transmission of next one. However, the impact of feedback delay has not been considered and analyzed, or robustness of the proposed solutions to feedback delay is not studied. Consequently, an important question is how the transmitter should proceed with the probing before receiving the user feedback, i.e., how to use the time spent on the delay efficiently for BA so as to reduce the BA overhead which may be critical for latency-sensitive applications.

In this paper, we investigate the problem of single-user interactive BA in the presence of feedback delay for the probing packets. During BA, the BS transmits $b$ BA packets and each packet has a feedback delay of $d$ time-slots. We assume a single cluster channel (either LoS or NLoS) model for the user and consider analog beamforming at the BS while the user is assumed to have an omnidirectional reception pattern. We also assume that the probing packets as well as their feedback can be decoded correctly at the user and the BS, respectively (error free decoding). Additionally, to avoid the problem of having arbitrary fragmented beams which may not be feasible in practice, we assume that the BS is constrained to use contiguous beams\footnote{Contiguous beams are the beams with a single main lobe covering one contiguous angular interval.}. Practical implementation of contiguous beams are less complex compared to non-contiguous (fragmented) beams with multiple angular coverage regions. From the BS perspective, we define the angular region containing the AoD of the user's channel as the uncertainty region (UR). After BA and during data communication, the BS's beam covers the entire UR to ensure reliable communication. Our objective is to maximize the expected beamforming gain during data communication which is inversely proportional with the width of the UR.
A summary of our contributions are as follows:

\begin{itemize}
    
    \item We view the BA problem with feedback delay as a source coding problem in which the BS asks $b$ yes/no questions and there is a constant delay of $d$ time-slots between each question and its answer. A question corresponds to asking if the AoD belongs to a certain angular interval and is implemented at the BS by sending a probing packet that covers the particular angular interval. Using this analogy, we show that the BA problem is equivalent to finding cyclically ordered sets (i.e., loops) \cite{novak1982cyclically} of angular intervals and feedback sequences referred to by \emph{loop of component beams} and \emph{$(b,d)$-unimodal loop}, respectively (Section~\ref{sec:BA}). 

    \item We investigate the properties of $(b,d)$-unimodal loops and provide a construction for $(b,d)$-unimodal loops that allow us achieve optimal performance in terms of the expected beamforming gain. Moreover, we show that the use of optimal contiguous scanning beams leads to contiguous URs at the end of BA when $d\geq 2$. This suggests optimal data transmission beams are also contiguous (Section~\ref{sec:unimodal}). 
    
    \item We provide a procedure for obtaining the optimal loop of component beams along with optimality achieving BA strategy for any arbitrary prior probability distribution on the AoD.  We also derive a tight lower bound on the maximum expected beamforming gain after BA~(Section~\ref{sec:angular_Int}).

    \item Through numerical evaluations and simulations, we investigate the impacts of feedback delay and AoD prior distribution on the optimal expected beamforming gain as well as the corresponding data beams. Furthermore, we compare the performance of the proposed optimal method with state-of-the-art methods in terms of the number of required time-slots to achieve a certain expected beamforming gain for the data beam. The results reveal that our method significantly improves the performance with respect to the state-of-the-art interactive and non-interactive BA (Section~\ref{sec:numerical}).
\end{itemize}

\section{System Model and Preliminaries} 
\label{sec:sys}
In this section, we outline general system assumptions (Sections \ref{sec:sys1} and \ref{sec:sys2}) and then provide the mathematical formulation of the problem (Sections \ref{subsec:prb} and \ref{sec:sys4}).

\subsection{Network Model}
\label{sec:sys1}
We consider a single-user downlink communication scenario in a single-cell mmWave system. Motivated by previous works (e.g., \cite{michelusi2018optimal,chiu2019active,khalili2021single}) and experimental studies (e.g., \cite{akdeniz2014millimeter}), we assume that the mmWave propagation channel has only a single dominant cluster. We denote the AoD corresponding with this cluster by $\psi$ which is unknown to the BS. Motivated by several prior works (e.g.,  \cite{khalili2021single,chiu2019active,hussain2017throughput,michelusi2018optimal}), we consider an analog BF architecture at the BS enabling it to use one beam at a time while the user is assumed to have an omnidirectional reception pattern. The goal of BA is to find the shortest possible angular interval, called uncertainty region (UR), containing $\psi$. This, in turn, results in the largest beamfomring gain. After BA, the BS will use a beam covering the entire UR for data transmission. We assume $\Psi \sim f_{\Psi} (\psi)$ for $\psi \in(0,2\pi]$ which accounts for the prior knowledge on the AoD (e.g., corresponding to the history of previously localized AoD in beam tracking applications). The BA algorithm can exploit this knowledge to reduce the width of UR and increase the beamforming gain during data transmission.

Due to practical constraints, we only consider use of contiguous beams as in \cite{khalili2020optimal}. Similar to \cite{khalili2021single,chiu2019active,hussain2017throughput,michelusi2018optimal,bai2015coverage}, we assume that the beams are ideal and use the \textit{sectored antenna} model introduced in \cite{ramanathan2001performance}. In this model, each beam is characterized by an angular interval $\Phi$ representing the angular coverage region (ACR) of the beam as well as a constant beamforming gain equal to $2\pi/|\Phi|$ representing the directivity gain of the beam\footnote{In this paper, we neglect the impact of side-lobes and leave it for future studies.}. In the case of contiguous beams, ACR $\Phi$ is a contiguous interval within $(0,2\pi]$. This model is justified as the BSs are envisioned to use large antenna arrays allowing for close to ideal beam shapes \cite{mmWave-survey-nyu}. 

\begin{figure}[t]
\centering
\includegraphics[width=0.4\linewidth, draft=false]{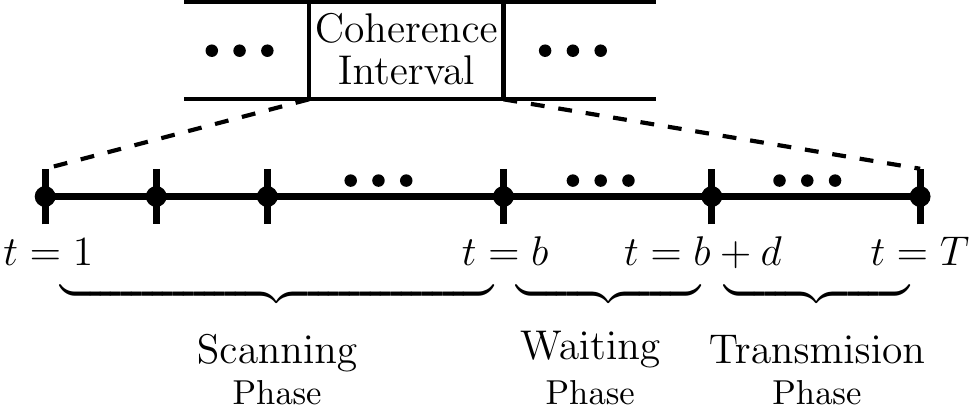}
\caption{Time-slotted system.}
\label{fig:sys_model}
\end{figure}

\subsection{Frames and Feedback}
\label{sec:sys2}

We consider an interactive BA scenario in which the BS receives feedback form the user during the transmission of BA probing packets and can change the subsequent scanning beams based on the feedback. Unlike conventional interactive BA in which the feedback to each transmitted packet is assumed to be available at the BS instantly \cite{michelusi2018optimal,khalili2021single}, we consider a fixed known delay of $d$ time-slots for each user feedback. This delay accounts for practical constraints such as processing and decoding delays at the transceivers or constraints imposed by the standard (e.g., 3GPP) where there are certain timings for various procedures. If an accurate value of this delay is not available, an upper bound can be utilized for our analysis. We assume that the feedback to each probing packet is either an acknowledgement (ACK) implying that the packet was received by the user or a negative acknowledgement (NACK) which indicates the user did not receive the packet. In our setup, NACK corresponds to the absence of ACK. Similar to \cite{michelusi2018optimal}, we consider that the feedback is received through an  error free dedicated control channel \cite{mmWave-survey-nyu}. Also, as in  \cite{hussain2017throughput,khalili2020optimal}, we assume that each BA packet is detected at the user without error if the beam used for BA packet transmission includes $\psi$.

Motivated by the above discussion, we consider a time-slotted system in which a fixed AoD is associated with user's channel over a coherence interval of duration $T$ time-slots. We assume that the coherence interval comprises three phases as shown in Fig.~\ref{fig:sys_model}. The first is \textit{scanning phase} in which the BS transmits $b$ BA probing packets through a set of scanning beams to scan the surrounding angular space. 
Since the response to each packet takes $d$ time-slots, after the scanning phase there is a \textit{waiting phase} in which the BS waits to receive the feedback to all of the probing packets. This phase lasts for $d$ time-slots and can be used for example, for data transmission to other users for which the BS has already performed BA. Subsequently, the BS processes all the feedback received and infers a beam for data communication. The rest of the coherence interval, i.e., the last $T-b-d$ time-slots, is called \textit{transmission phase} which is used for data transmission using the beam inferred from the previous phases.

Our main goal is \textit{i)} to design a dictionary of beams to be used in the scanning phase given the prior distribution of user's AoD, \textit{ii)} to provide an approach for using the designed beams in the scanning phase, and \textit{iii)} to propose a method for processing the user's feedback and construct a beam for data transmission. The main objective in all the steps is to maximize the expected beamforming gain during data transmission.  

\subsection{Scanning Beam Set and Data Beam}
\label{subsec:prb}

In the scanning phase, the BS uses $b$ scanning beams $\{\Phi_i\}_{i\in[b]}$ to transmit $b$ BA probing packets\footnote{We use the notation $[n]$ to represent the set $\{1,2,\ldots,n\}$.}. Let $a_i \in \{0,1\}$ denote the feedback received for the $i^{\rm th}$ BA packet, where $a_i = 1$ if ACK was received for $\Phi_i$ and $a_i = 0$, otherwise. In general, the scanning beam $\Phi_i$ is determined based on the received feedback sequence until time-slot $i$ i.e., $(a_1,a_2, \ldots, a_{i-d})$. To model this, we adopt a hierarchical beam set $\Sset(b,d) = \{\Sset_i\}_{i \in [b]}$, where $\Sset_i = \{S_{i,m}\}_{m\in[M(i-d,d)]}$ denotes the set of designed scanning beams for time-slot $i$. Note that there are a total of $M(i-d,d)\leq 2^{i-d}$ possible feedback sequences until time-slot $i$. To elaborate, the set $\Sset_i$ contains a beam for each possible feedback sequence received by the $i^{\rm th}$ time-slot, the BS selects the beam $S_{i,m}\in \Sset_i$ for the transmission of the next probing packet based on the realization of the received feedback sequence. To simplify the notation, unless necessary, we refer to $\Sset(b,d)$ by $\Sset$.

Next, we describe how to process the received user feedback and infer the shortest angular region including the user's AoD (i.e., UR). Given an AoD realization $\psi$, we denote the angular region that the BS chooses for data transmission (i.e., UR) by $\Beam(\Sset,\psi)$. Under the assumption of having a single dominant path channel and error free system, the minimum width UR is 
\begin{align}
\label{eq:phis}
\Beam(\Sset,\psi) = \cap_{i = 1}^b \Theta(\Phi_i,a_i),
\end{align}
where $\Theta(\Phi_i,a_i) = \Phi_i$ if $a_i = 1$ which corresponds to $\psi \in \Phi_i$, and $\Theta(\Phi_i, a_i) = (0,2\pi] - \Phi_i$, otherwise.
Based on \eqref{eq:phis}, given $\Sset$, we get an UR for each possible feedback sequence $(a_1, a_2, \ldots, a_b)$. We denote the set of possible URs by $\Uset = \{U_m\}_{m \in M(b,d)}$. Note that the number of URs at the end of the BA is the same as the number of possible feedback sequences received by time-slot $b+d$, that is $M(b,d)$. Note that the URs $U_m$s are disjoint. For reliable transmission, the BS forms a beam to cover the entire UR in the data transmission phase which means that $\Beam(\Sset,\psi) = U_m$ if $\psi \in U_m$. To further clarify the above discussions, let's consider the following example.

\begin{figure*}[t!]
    \centering
    \includegraphics[width = 0.7\textwidth]{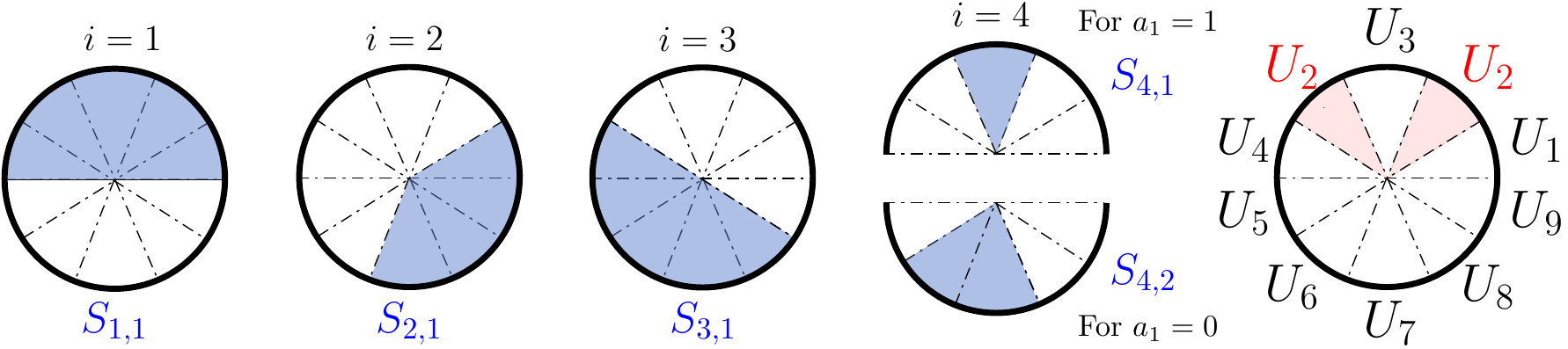}
     \caption{A set of scanning beams and their corresponding uncertainty regions for $b=4$ and $d=3$.}
    \label{fig:SCLI}
\end{figure*}

\begin{Example}
\label{ex:Iset}
Fig.~\ref{fig:SCLI} illustrates a possible set of scanning beams for $b=4$ and $d=3$. In this case, $\Sset(4,3) = \{\Sset_1,\Sset_2,\Sset_3,\Sset_4\}$, with $\Sset_i = 
\{\Phi_i\}$ for $i \in \{1,2,3\}$ each consisting of a single possible scanning beam as no feedback is received prior to fourth time-slot. However, at the fourth time-slot, we receive the feedback to the beam $\Phi_1$ and so there are two possibilities for $\Phi_4$. Here,  we have $\Sset_4 = \{S_{4,1}, S_{4,2}\}$. As shown in Fig.~\ref{fig:SCLI}, the set $\Sset(4,3)$ creates the URs $\Uset = \{U_j\}_{j=1}^9$.
\end{Example}
One important observation from Ex. \ref{ex:Iset} is that even though the scanning beams are contiguous, the URs might be fragmented which is the case for $U_2$. However, for our optimality achieving BA procedure, which will be discussed later in the paper, we show that the resulting URs are contiguous intervals when using contiguous scanning beams. This means that the data beam will also be contiguous.

\subsection{Problem Formulation}
\label{sec:sys4} 
The objective is to design the set of scanning beams $\Sset$ so as to maximize the expected beamforming gain during data transmission. Based on the discussions in Sec.~\ref{sec:sys1}, this is formulated as follows.
\begin{align}
\label{eq:optimization} 
\begin{aligned}
\Sset^* = \argmax_{\Sset} \mathbb{E}_{\Psi}\left[\frac{2\pi}{|\Beam(\Sset,\Psi)|}\right],
\end{aligned}
\end{align}
where the expectation is taken over the AoD prior $f_{\Psi}(\psi)$. Using the definition of URs, we can rewrite the expected beamforming gain as 
\begin{align}
\label{eq:optimization2}
\mathbb{E}_{\Psi}\left[\frac{2\pi}{|\Beam(\Sset,\Psi)|}\right] = \sum_{m=1}^{M(b,d)} \frac{2\pi}{|
U_m|}\int_{\psi\in U_m}\!\!\!\!f_{\Psi}(\psi) d\psi,
\end{align}
where notation $|U_m|$ is the Lebesgue measure of $U_m$, which is equal to the total width of the intervals in the case where $U_m$ is the union of a finite number of intervals. 

As it can be observed form \eqref{eq:optimization2}, the performance of any scanning beam set $\Sset$ depends solely on the set of possible URs it may generate. Therefore, we can characterize the performance of the optimal scanning beam set by investigating the properties of the corresponding URs. 

\section{Beam Alignment and Cyclic Sets}
\label{sec:BA}
To find the optimal scanning beam set $\Sset^*$, we use an information theoretic approach in which we view the discussed BA problem as a source coding problem. The BS asks $b$ questions whose answers (the feedback sequences) represent the source codewords describing the user's data beam. Unlike a finite alphabet source coding problem, here, the questions are intervals inside $(0,2\pi]$. Using this framework, we reduce the optimization in \eqref{eq:optimization2} into finding a pair of cyclically ordered sets (i.e., loops) \cite{novak1982cyclically} of angular intervals and feedback sequences as discussed below. 

\subsection{Preliminaries and Definitions}
\label{subsec:prelem}
Considering the example shown in Fig.~\ref{fig:SCLI}, we observe that the angular intervals in between the dotted-dashed lines can be used as basis for representation of the scanning beams. We refer to such angular intervals as \textit{component beams} formally defined below.
\begin{Definition}[\textbf{Component Beams}]
\label{def:cb}
Given a set of scanning beams $\Sset$, we sort the endpoints of the scanning beams and remove the repetitions. We define each angular interval in between two consecutive endpoints as a component beam.
\end{Definition}
Note that based on Def.~\ref{def:cb}, the component beams are contiguous and partition interval $(0,2\pi]$. As a result, we can use their position on the circle to define a cyclic order and form a loop of component beams. We denote this loop by $\Iset$. For the setup in Ex.~\ref{ex:Iset} these component beams and their loop are\footnote{We use the notation $\odot[\ldots]$ to indicate loops.}
\begin{align}
\begin{aligned}
\tikzset{every picture/.style={line width=0.75pt}} 
\begin{tikzpicture}[x=0.75pt,y=0.75pt,yscale=-1,xscale=1]
\path (400,50); 
\draw   (146.4,110.9) .. controls (146.4,92.4) and (161.4,77.4) .. (179.9,77.4) .. controls (198.4,77.4) and (213.4,92.4) .. (213.4,110.9) .. controls (213.4,129.4) and (198.4,144.4) .. (179.9,144.4) .. controls (161.4,144.4) and (146.4,129.4) .. (146.4,110.9) -- cycle ;
\draw [dash dot] [line width=0.5]    (153.28,91.56) -- (186.71,115.85) -- (206.52,130.24) ;
\draw [dash dot] [line width=0.5]     (169.73,79.61) -- (190.07,142.19) ;
\draw  [dash dot] [line width=0.5]   (190.07,79.61) -- (169.73,142.19) ;
\draw  [dash dot] [line width=0.5]   (147,110.9) -- (212.8,110.9) ;
\draw  [dash dot] [line width=0.5]  (153.28,130.24) -- (206.52,91.56) ;
\draw (22,142.4) node [anchor=north west][inner sep=0.75pt]    {$\ $};
\draw (211.6,89.4) node [anchor=north west][inner sep=0.75pt]    {$I_{1}$};
\draw (198.6,68.4) node [anchor=north west][inner sep=0.75pt]    {$I_{2}$};
\draw (173.6,59.4) node [anchor=north west][inner sep=0.75pt]    {$I_{3}$};
\draw (146.6,68.4) node [anchor=north west][inner sep=0.75pt]    {$I_{4}$};
\draw (132.6,89.4) node [anchor=north west][inner sep=0.75pt]    {$I_{5}$};
\draw (132.6,113.4) node [anchor=north west][inner sep=0.75pt]    {$I_{6}$};
\draw (147.6,135.4) node [anchor=north west][inner sep=0.75pt]    {$I_{7}$};
\draw (172.73,144.59) node [anchor=north west][inner sep=0.75pt]    {$I_{8}$};
\draw (196.6,136.4) node [anchor=north west][inner sep=0.75pt]    {$I_{9}$};
\draw (209.6,117.4) node [anchor=north west][inner sep=0.75pt]    {$I_{10}$};
\draw (250,100) node [anchor=north west][inner sep=0.75pt]    {$\ \ \ \ \mathcal{I} \ =\odot \left[
I_1,I_2, I_3, I_4, I_5, I_6, I_7, I_8, I_9, I_{10}\right]$};
\end{tikzpicture}
\qquad \qquad 
\end{aligned}
\label{eq:Iset_ex}
\end{align}

We can use $\Iset$ to provide a binary loop representation of each of the scanning beams. For a given scanning beam, we construct this binary loop by putting one at the position of the component beams which represent (i.e., partition) the scanning beam and zero elsewhere. For example, the binary loop representing the beam $S_{2,1}$ in Ex.~\ref{ex:Iset} is $\odot[1~0~0~0~0~0~0~1~1~1]$.
As we can see, the positions of ones in this loop are consecutive (note that the first element is located right after the last element since it is a loop). This property holds for all of the scanning beams since they are contiguous.
We refer to such binary loops as \textit{unimodal loops} formally defined below. 

\begin{Definition}[\textbf{Unimodal Loop}]
\label{def:ul}
A unimodal loop is a binary loop in which the positions of ones (if any) are consecutive.
\end{Definition}

Next, we use the loop of component beams and the unimodality property to study some properties of feedback sequences.

\subsection{Scanning Beams Set Decomposition}

Let us create a loop of the feedback sequences of a scanning beam set $\Sset$ by first forming its loop of component beams and then replacing each element of the loop with its corresponding feedback sequence, i.e., the feedback sequence that the BS would receive if the AoD of the user falls inside that component beam. We denote this loop by $\Lset(b,d)$, where $b$ is the BA duration and $d$ is the delay. As an example, for the setup in Ex.~\ref{ex:Iset} whose loop of component beams is shown in \eqref{eq:Iset_ex}, this loop is
\begin{align}
    \mathcal{L}(4,3) \ =\odot \left[\def\arraystretch{0.5}\begin{array}{ c c c c c c c c c c }
1 & 1 & 1 & 1 & 1 & 0 & 0 & 0 & 0 & 0\\
1 & 0 & 0 & 0 & 0 & 0 & 0 & 1 & 1 & 1\\
0 & 0 & 0 & 0 & 1 & 1 & 1 & 1 & 1 & 0\\
0 & 0 & 1 & 0 & 0 & 0 & 1 & 1 & 0 & 0
\end{array}\right].
\label{eq:IandL}
\end{align}
To simplify the notation, unless necessary, we refer to $\Lset(b,d)$ by $\Lset$. Before, we provide some general properties of $\Lset$, let us first consider the following example.  
\begin{Example}
\label{ex:du}
Consider the setup in Ex.~\ref{ex:Iset} and its loop of feedback sequences, $\Lset(4,3)$, shown in \eqref{eq:IandL}. The first row (i.e., the first feedback bit $a_1$) corresponds to the first scanning beam $\Phi_1$. More precisely, the first bit of a column is one if its component beam is included in $\Phi_1$ and zero, otherwise. As a result, the first row of the loop $\Lset(4,3)$ is the binary loop representation of $\Phi_1 = S_{1,1}$ in terms of the component beams, which, as discussed in Sec.~\ref{subsec:prelem}, is unimodal. Similarly, we can show that the second and third rows are also unimodal loops since they are binary representations of the beams $S_{2,1}$ and $S_{3,1}$, respectively. However, the fourth row is not unimodal. The reason is that this row corresponds to two beams, namely $S_{4,1}$ and $S_{4,2}$ depending on the feedback $a_1$ which is received at the fourth time-slot. If $a_1 = 1$, the BS uses the beam $S_{4,1}$, otherwise it uses the beam $S_{4,2}$. Therefore, the forth bit of each column is one if its component beam falls inside $\phi_1 \cap S_{4,1}$ or $((0, 2\pi] - \phi_1) \cap S_{4,2}$ and zero, otherwise. Nevertheless, if we consider the columns of the loop $\Lset(4,3)$ corresponding to $a_1 = 1$ and $a_1 = 0$ separately.
\begin{align}
\label{eq:ex01}
    \tilde{\Lset}_1(4,3) \ =\odot \left[\def\arraystretch{0.5}\begin{array}{ c c c c c c c c c c }
1 & 1 & 1 & 1 & 1 \\
1 & 0 & 0 & 0 & 0 \\
0 & 0 & 0 & 0 & 1 \\
0 & 0 & 1 & 0 & 0
\end{array}\right], \quad
   \tilde{\Lset}_2(4,3) \ =\odot \left[\def\arraystretch{0.5}\begin{array}{ c c c c c c c c c c }
0 & 0 & 0 & 0 & 0\\
0 & 0 & 1 & 1 & 1\\
1 & 1 & 1 & 1 & 0\\
0 & 1 & 1 & 0 & 0
\end{array}\right],
\end{align}
We observe that the resulting binary loop of the fourth row in the sub-loop $\tilde{\Lset}_1(4,3)$ ($\tilde{\Lset}_2(4,3)$), namely $\odot [0$ $0$ $1$ $0$ $0]$ ($\odot[0$ $1$ $1$ $0$ $0]$), is unimodal\footnote{A sub-loop of a loop is a loop in which some of the elements of the original loop are removed.}. 

Another observation is that the loop $\Lset(4,3)$ does not have any consecutive repetitions of columns (feedback sequences). This stems from the component beams definition in Def.~\ref{def:cb}.
\end{Example}

Generalizing the above example, we have the following theorem. 
\begin{Theorem} [\textbf{Scanning Beam Set Decomposition}]
\label{thm:Lprop}
A scanning beam set $\Sset$ can be decomposed into a loop of component beams $\Iset$ and a loop of feedback sequences $\Lset$, where column $j$ of $\Lset$ is the feedback sequence corresponding to column $j$ of $\Iset$. The loop $\Lset$ has the following properties.
\begin{enumerate}
    \item For $i\in [b]$ and for each prefix of length $i-d$, the loop consisting of the $i^{\rm th}$ bits of the feedback sequences with that prefix is unimodal. Note that for $i\leq d$, we assume all feedback sequences have the same prefix.
    \item It does not have any consecutive column repetitions.
\end{enumerate}
We refer to any loop of feedback sequences satisfying the above two properties as \emph{$(b,d)$-unimodal}.
\end{Theorem}
\begin{proof}
The proof is provided in Appendix~\ref{app:Lprop}.
\end{proof}

While the loop of feedback sequences, $\Lset$, cannot have consecutive column repetitions, non-consecutive repetitions are still possible. For example, in the loop $\Lset(4,3)$ shown in \eqref{eq:IandL} columns two and four are identical. On the other hand, the UR corresponding to a feedback sequence is the union of component beams in $\Iset$ that have that feedback sequence. Therefore, column repetition in $\Lset$, means there are non-contiguous URs. For example, for the loop $\Lset(4,3)$ in \eqref{eq:IandL}, the UR $U_2 = \cup\{I_2,I_4\}$ is non-contiguous. This can also be observed from Fig.~\ref{fig:SCLI}.

\subsection{Scanning Beam Set Construction}

So far, we have shown that a scanning beam set $\Sset$ can be decomposed into the pair $(\Iset,\Lset)$, where $\Lset$ is $(b,d)$-unimodal. Conversely, in this sub-section, we argue that given a pair $(\Iset,\Lset)$, where $\Iset$ and $\Lset$ have the same cardinality (number of columns) and $\Lset$ is $(b,d)$-unimodal, we can construct a scanning beam set that corresponds to $(\Iset,\Lset)$. Before providing the construction, let us consider the following example. 

\begin{Example}
\label{ex:SCLI}
Assume that the loop of component beams $\Iset$, and feedback sequences, $\Lset$, are given as in \eqref{eq:Iset_ex} and \eqref{eq:IandL}, respectively, and we want to find a scanning beams set $\Sset$ that leads to the pair $(\Iset,\Lset)$. Note that the loop $\Lset(4,3)$ is $(4,3)$-unimodal. Since $d=3$ and no feedback is received prior to the fourth time-slot, the set $\Sset_i$ for $i\leq3$, only consists of one scanning beam. This scanning beam is the union of the component beams in $\Iset$ whose corresponding bit in the $i^{\rm th}$ row of the loop $\Lset(4,3)$ is one. At the fourth time-slot, however, we receive the feedback to the first scanning beam (i.e., $a_1$). Since, $\Lset$ has feedback sequences with both $a_1 = 1$ and $a_1 = 0$, $\Sset_4$ will have two scanning beams, $S_{4,1}$ for $a_1= 1$ and $S_{4,2}$ for $a_1 = 0$.

Let us consider the beam $S_{4,1}$. The component beams that correspond to $a_1 = 1$ (component beams whose feedback sequences in the loop $\Lset(4,3)$ have $a_1 = 1$) are $I_1,~I_2,~I_3,~I_4$, and $I_5$ and let us consider the sub-loop of $\Lset$ corresponding to these component beams,
\begin{align}
\hat{\Lset}_1 =\odot\left[\def\arraystretch{0.5}\begin{array}{ c c c c c}
1 & 1 & 1 & 1 & 1 \\
1 & 0 & 0 & 0 & 0 \\
0 & 0 & 0 & 0 & 1 \\
0 & 0 & 1 & 0 & 0
\end{array}\right].
\end{align}

On the other hand, a component beam is included in $S_{4,1}$ only if its corresponding fourth feedback bit (fourth row of the loop $\Lset(4,3)$) is one. This suggests that the beam  $S_{4,1}$ includes the component beam $I_3$, but not $I_1~,~I_2,~I_4$, and $I_5$. Since the scanning beams are contiguous, we have $S_{4,1} = I_3$. Similarly, one can argue that the beam $S_{4,2}$ includes the component beams $I_7$ and $I_8$, but not  $I_6,~I_9$, and $I_{10}$ which leads to $S_{4,2} = I_7 \cup I_8$. This construction leads to the scanning beam set in Fig.~\ref{fig:SCLI}.

For the considered loop of the feedback sequences, the loop $\Lset(4,3)$, $S_{4,1}$ and $S_{4,2}$ are unique. However, this is not the case for any $(4,3)$-unimodal loop $\Lset(4,3)$. To elaborate, suppose that we had the following sub-loop for $S_{4,1}$ 
\begin{align}
\hat{\Lset}_1 =\odot\left[\def\arraystretch{0.5}\begin{array}{ c c c c c}
1 & 1 & 1 & 1 & 1 \\
1 & 0 & 0 & 0 & 0 \\
0 & 0 & 0 & 0 & 1 \\
0 & 0 & 0 & 1 & 1
\end{array}\right].
\end{align}
Then, $S_{4,1}$ should have included the component beams $I_4$ and $I_5$, but not $I_1,~I_2,$ and $~I_3$. There are multiple contiguous beams that satisfy these constraints. For example, the beams $\cup\{I_4,I_5\}$, $\cup\{I_4,I_5,I_6\}$, or  $\cup\{I_4,I_5,I_6, I_7\}$ are all valid choices.
\end{Example}

Generalizing this example we have the following result.

\begin{Theorem}[\textbf{Scanning Beam Set Construction}]
\label{thm:1t1}
Given a loop of component beams and a $(b,d)$-unimodal loop of the same cardinality $(\Iset,\Lset)$, one can construct a scanning beam set $\Sset= \{\Sset_i\}_{i\in [b]}$ that corresponds to the pair $(\Iset, \Lset)$. The construction is described below.

\begin{Construction}
\label{cons:Prob_beam_set}
For the set $\Sset_i$, where $i\in[b]$, we construct a beam for each unique prefix of length $i-d$ as it follows. For $i\leq d$, we assume all feedback sequences have the same prefix.
\begin{enumerate}
    \item Create sub-loops of $\Lset$ with the given prefix. 
    \item Create the sets $\Iset_{\rm in}$ and $\Iset_{\rm out}$ including the component beams of these columns which have bit $1$ and $0$ as their $i^{\rm th}$ bit feedback sequence ($i^{\rm th}$ row of $\Lset$), respectively.
    \item The scanning beam corresponding to this prefix is any union of component beams in $\Iset$ that has the following properties: i) It is  contiguous, ii) it does not include the component beams in $\Iset_{\rm out}$, and iii) it includes all of the component beams in $\Iset_{\rm in}$.
\end{enumerate}
\end{Construction}

\end{Theorem}
\begin{proof}
The proof is provided in Appendix~\ref{app:1t1}, where we show that the resulting scanning beam from Const.~\ref{cons:Prob_beam_set} leads to the pair $(\Iset, \Lset)$.
\end{proof}

Using Theorems \ref{thm:Lprop} and \ref{thm:1t1}, we can establish the equivalence of finding the optimal scanning beam set $\Sset\str$ in \eqref{eq:optimization}, and finding the optimal pair of loops of component beams and loops of feedback sequences $(\Iset\str, \Lset\str)$. 
We characterize the properties of $(b,d)$-unimodal loops and construct an optimality achieving loop of feedback sequences $\Lset\str$ in the next section. Then, in Section \ref{sec:angular_Int}, we provide an optimal loop of component beams $\Iset\str$ and consecutively an optimal scanning beam set $\Sset\str$, and a lower bound on the maximum expected beamforming gain.

\section{Optimal Unimodal Loop}
\label{sec:unimodal}

In this section, we consider a particular family of $(b,d)$-unimodal loops for which we provide a construction and derive their properties. Then, we show that it is sufficient to only consider this family of $(b,d)$-unimodal loops in order to find an optimal loop of feedback sequences, $\Lset\str$. We start with the necessary definitions.
\subsection{Preliminaries and Definitions}

As discussed in Sec.~\ref{sec:BA}, given a pair $(\Iset,\Lset)$, where $\Iset$ and $\Lset$ have the same number of elements and $\Lset$ is $(b,d)$-unimodal, we can use Const.~\ref{cons:Prob_beam_set} in Thm.~\ref{thm:1t1} to generate a scanning beam set $\Sset$ corresponding to $(\Iset,\Lset)$. Moreover, each UR of $\Sset$ is the union of the component beams in $\Iset$ that have the same feedback sequences (same columns of $\Lset$).

Consider the optimization in \eqref{eq:optimization2}. For the case of the uniform prior, given that the number of URs (number of unique columns in $\Lset$), $M(b,d)$, is fixed, it is straightforward to see that to maximize the beamforming gain we need to use a loop of component beams, $\Iset$, that leads to equal width URs. As a result, we get the optimal beamforming gain of $M(b,d)$. Therefore, to achieve the optimal performance for the case of the uniform prior, we need to use a $(b,d)$-unimodal loop that has the maximum number of unique columns.  On the other hand, for a general prior $f_{\Psi}(\cdot)$, the more columns we have in the loop of feedback sequences, $\Lset$, the more degrees of freedom we will have in optimizing the URs. Motivated by these observations, we formally define a new class of $(b,d)$-unimodal loops below. 

\begin{Definition}[\textbf{Max Cardinality Loop}]
\label{def:MCL}
A max cardinality loop $\MCLset(b,d)$ is a $(b,d)$-unimodal loop that has \textit{i)} the maximum number of columns and \textit{ii)} the maximum number of unique columns (i.e., number of possible feedback sequences) among all possible $(b,d)$-unimodal loops. To simplify the notation, unless necessary, we refer to $\MCLset(b,d)$ by $\MCLset$.
\end{Definition}

In the sequel, we will show that max cardinality loop, $\MCLset$, exists by providing a construction and prove that it is sufficient to use an $\MCLset$ to find an optimal scanning beam set, $\Sset\str$. For our construction, we make use of the following observation.

Let us consider the loop of feedback sequences in \eqref{eq:IandL} corresponding to the scanning beam set $\{\Sset_i\}_{i\in[4]}$ shown in Fig.~\ref{fig:SCLI}. If we remove the last row of this loop and then the resulting consecutive repetitions of its columns, we get the loop
\begin{align}
\label{eq:L3}
\Lset(3,3) = \odot \left[\def\arraystretch{0.5}\begin{array}{ c c c c c c c c c c }
1 & 1 & 1 & 0 & 0 & 0\\
1 & 0 & 0 & 0 & 1 & 1\\
0 & 0 & 1 & 1 & 1 & 0
\end{array}\right]
\end{align}
which is the loop of feedback sequences corresponding to the scanning beam set $\{\Sset_i\}_{i\in[3]}$. In general, given a scanning beam set $\{\Sset_i\}_{i\in[b]}$ and its loop of feedback sequences $\Lset$, by removing the last row of  $\Lset(b,d)$ and the resulting consecutive repetitions, we get the loop of feedback sequences $\Lset(b-1,d)$ corresponding to  $\{\Sset_i\}_{i\in[b-1]}$. Motivated by this observation, we define a parent-child hierarchy for the $(b,d)$-unimodal loops formally defined below.

\begin{Definition}[\textbf{Parent and Child Loops}]
\label{def:parent}
For a $(b,d)$-unimodal loop $\Lset(b,d)$, the loop $\Lset(b-i,d)$ obtained by removing the last $i$ rows of $\Lset(b,d)$ and then removing the consecutive column repetitions is defined as the \emph{parent loop of order $i$}. The parent loop of order one is called the \emph{parent loop}. Conversely, the $\Lset(b,d)$ is the child loop of $\Lset(b-1,d)$.
\end{Definition}

Note that the unique parent loop of $(b,d)$-unimodal loop is $(b-1,d)$-unimodal. However, a parent loop can result in different child loops. For our results, we also need to define the following notation related to the unimodal loops in Def.~\ref{def:ul}.

\begin{Definition}[\textbf{First and Last Flip Positions}]
\label{def:flfp}
Consider a unimodal binary loop of a fixed length. This loop has two important positions, the position of i) the first bit which is flipped, denoted by $i_f$, and ii) the position of the bit whose next bit is flipped again, denoted by $i_l$. The case where the loop is all ones or zeros is denoted by $i_f = i_l = \infty$. We refer to $i_f$ and $i_l$ by the first and last flip positions, respectively.  
\end{Definition}

As an example, the loop $\odot[0~0~0~1~0~0~0]$ is a unimodal loop of length seven with $i_f = i_l = 4$.

\subsection{Child Loop Construction and Properties}

To construct an $\MCLset$, we first study the properties of the parent-child hierarchy in Def.~\ref{def:parent} and discuss child loop construction form a parent loop. To this end, let us consider the following example.

\begin{Example}
\label{ex:const_p1}
Consider the loop $\Lset(4,3)$ and its parent loop $\Lset(3,3)$ in \eqref{eq:IandL} and \eqref{eq:L3}, respectively. We can generate $\Lset(4,3)$ from $\Lset(3,3)$ as follows. 
Viewing each column as a feedback sequence, and creating the sub-loops of columns of $\Lset(4,3)$ consisting of columns with the same prefixes of length one allows us to consider each possible scanning beam used at the fourth time-slot separately. So, let us consider the sub-loops of $\Lset(3,3)$ consisting of columns with the same prefixes of length one.
We get
\begin{align}
\label{eq:exc0}
    \tilde{\Lset}_1(3,3) =\odot \left[\def\arraystretch{0.5}\begin{array}{ c c c }
1 & 1 & 1 \\
1 & 0 & 0 \\
0 & 0 & 1
\end{array}\right],\qquad
\tilde{\Lset}_2(3,3) =\odot \left[\def\arraystretch{0.5}\begin{array}{ c c c}
0 & 0 & 0\\
0 & 1 & 1\\
1 & 1 & 0
\end{array}\right].
\end{align}

Since $\Lset(4,3)$ should be $(4,3)$-unimodal, it cannot have more than three consecutive columns with the same prefix of length three otherwise, it will not follow Def.~\ref{def:ul}. We will prove this rigorously as part of the proof of Thm.~\ref{thm:gcc}. Now, let us repeat each column in the sub-loops of the loop $\Lset(3,3)$ three times consecutively. We have
 \begin{align}
\label{eq:exc1}
    \widehat{\Lset}_1(3,3) =\odot \left[\def\arraystretch{0.5}\begin{array}{ c c c c c c c c c}
1 & 1 & 1 & 1 & 1 & 1 & 1 & 1 & 1\\
1 & 1 & 1 & 0 & 0 & 0 & 0 & 0 & 0 \\
0 & 0 & 0 & 0 & 0 & 0 & 1 & 1 & 1
\end{array}\right],\\
\widehat{\Lset}_2(3,3) =\odot \left[\def\arraystretch{0.5}\begin{array}{ c c c c c c c c c}
0 & 0 & 0 & 0 & 0 & 0 & 0 & 0 & 0\\
0 & 0 & 0 & 1 & 1 & 1 & 1 & 1 & 1\\
1 & 1 & 1 & 1 & 1 & 1 & 0 & 0 & 0
\end{array}\right].
\end{align}

Considering the counterpart BA problem, the column repetition is to account for the possible addition of the component beams when adding new beams to the scanning beam set. Now, we add a row to each of these sub-loops. These rows, which are unimodal correspond to the scanning beams used at the fourth time-slot. In this example, we use the rows $\Pset_1 = \odot [0$ $0$ $0$ $0$ $1$ $0$ $0$ $0$ $0]$ and  $\Pset_2 = \odot [0$ $0$ $1$ $1$ $0$ $0$ $0$ $0$ $0]$ to $\widehat{\Lset}_1(3,3)$ an $\widehat{\Lset}_3(3,3)$, respectively. We get

\begin{align}
\label{eq:exc2}
    \overline{\Lset}_1(4,3) =\odot \left[\def\arraystretch{0.5}\begin{array}{ c c c c c c c c c}
1 & 1 & 1 & 1 & 1 & 1 & 1 & 1 & 1\\
1 & 1 & 1 & 0 & 0 & 0 & 0 & 0 & 0 \\
0 & 0 & 0 & 0 & 0 & 0 & 1 & 1 & 1 \\
0 & 0 & 0 & 0 & 1 & 0 & 0 & 0 & 0
\end{array}\right],\\
\overline{\Lset}_2(4,3) =\odot \left[\def\arraystretch{0.5}\begin{array}{ c c c c c c c c c}
0 & 0 & 0 & 0 & 0 & 0 & 0 & 0 & 0\\
0 & 0 & 0 & 1 & 1 & 1 & 1 & 1 & 1\\
1 & 1 & 1 & 1 & 1 & 1 & 0 & 0 & 0\\
0 & 0 & 1 & 1 & 0 & 0 & 0 & 0 & 0
\end{array}\right].
\end{align}

Next, to get the child loop, we need to combine the sub-loops $\overline{\Lset}_1(4,3)$ and $\overline{\Lset}_2(4,3)$ into one loop while preserving the order of the columns inside each sub-loop and order of their prefixes of length three in the loop $\Lset(3,3)$. This step guarantees that the cyclic order of the columns in the loop $\Lset(4,3)$ is aligned with that of the loop $\Lset(3,3)$. The resulting loop is
\begin{align}
\label{eq:exc3}
    \underline{\Lset}_1(4,4) =\odot \left[\def\arraystretch{0.5}\begin{array}{ c c c c c c c c c c c c c c c c c c}
1 & 1 & 1 & 1 & 1 & 1 & 1 & 1 & 1 & 0 & 0 & 0 & 0 & 0 & 0 & 0 & 0 & 0\\
1 & 1 & 1 & 0 & 0 & 0 & 0 & 0 & 0 & 0 & 0 & 0 & 1 & 1 & 1 & 1 & 1 & 1 \\
0 & 0 & 0 & 0 & 0 & 0 & 1 & 1 & 1 & 1 & 1 & 1 & 1 & 1 & 1 & 0 & 0 & 0 \\
0 & 0 & 0 & 0 & 1 & 0 & 0 & 0 & 0 & 0 & 0 & 1 & 1 & 0 & 0 & 0 & 0 & 0
\end{array}\right].
\end{align}

In the BA problem, this specific way of combining the sub-loops makes sure that the order of the component beams after adding new scanning beams does not contradict with the order of the component beams before the addition. 
Finally, to get $\Lset(4,3)$, we remove the consecutive repetitions of the columns.
\end{Example}

Generalizing Ex.~\ref{ex:const_p1}, we have the following theorem. 

\begin{Theorem}[\textbf{Child 
Loop Construction}]
\label{thm:gcc}
Given a parent loop $\Lset(b-1,d)$, all possible $(b,d)$-unimodal child loops, $\Lset(b,d)$s, can be generated using the following construction.
\begin{Construction}
~
\label{con:child}
\begin{enumerate}
    \item Sub-loop formation: Form sub-loops of $\Lset(b-1,d)$ consisting of columns with the same prefix of length $b-d$. If $b\leq d$, then $\Lset(b-1,d)$ has one sub-loop which is itself. 
    \item Column repetition: Repeat the columns in each sub-loop three times consecutively.
    \item Loop concatenation: Add a unimodal row to each sub-loop.
    \item Ordered combination: Combine the sub-loops to form a bigger loop by preserving the order of the columns in each sub-loop and order of their prefixes of length $b-1$ in the parent loop $\Lset(b-1,d)$. 
    \item Repetition removal: Remove the consecutive repetitions of the columns. 
\end{enumerate}
\end{Construction}
\end{Theorem}
\begin{proof}
The proof is provided in Appendix~\ref{app:gcc}.
\end{proof}

Next, we will quantify the maximum number of columns and unique columns one can have in a child loop given a parent loop. Let us consider the following example first.

\begin{Example}
\label{ex:const_p}
Consider the parent loop of $\Lset(3,3)$ given in \eqref{eq:L3} which is
 \begin{align}
& {\Lset}(2,3) =\odot \left[\def\arraystretch{0.5}\begin{array}{ c c c c c c c c c}
1 & 1 & 0 & 0 \\
1 & 0 & 0 & 1
\end{array}\right].
 \end{align}
Let us use Const.~\ref{con:child} (Thm.~\ref{thm:gcc}) to create child loops of this parent loop. We first need to create sub-loops of the columns with same prefix of length $ 3-3 = 0 $ which means there is only one sub-loop. Then, we need to repeat each column three times (Steps 1 and 2). We get
 \begin{align}
& \widehat{\Lset}(2,3) =\odot \left[\def\arraystretch{0.5}\begin{array}{ c c c c c c c c c c c c}
1 & 1 & 1 & 1 & 1 & 1 & 0 & 0 & 0 & 0 & 0 & 0\\
1 & 1 & 1 & 0 & 0 & 0 & 0 & 0 & 0 & 1 & 1 & 1
\end{array}\right]
 \end{align}
 For Step~3 (loop concatenation), we consider three cases using the unimodal loops $\Pset_1 = \odot[1~1~1~0~0~0~0~0~0~0~0~0]$, $\Pset_2 = \odot[0~1~0~0~0~0~0~0~0~0~0~0]$, and $\Pset_3 = \odot[0~1~1~1~1~0~0~0~0~0~0~0]$. After Steps 4 and 5, for the considered unimodal loops, we get the child loops
 \begin{gather}
\Lset_1(3,3) =\odot \left[\def\arraystretch{0.5}\begin{array}{ c c c c c c c c c}
1 & 1 & 0 & 0\\
1 & 0 & 0 & 1\\
1 & 0 & 0 & 0
\end{array}\right],\quad
\Lset_2(3,3) =\odot \left[\def\arraystretch{0.5}\begin{array}{ c c c c c c c c c}
1 & 1 & 1 & 1 & 0 & 0 \\
1 & 1 & 1 & 0 & 0 & 1 \\
0 & 1 & 0 & 0 & 0 & 0 
\end{array}\right],\\
\Lset_3(3,3) =\odot \left[\def\arraystretch{0.5}\begin{array}{ c c c c c c c c c}
1 & 1 & 1 & 1 & 0 & 0 \\
1 & 1 & 0 & 0 & 0 & 1 \\
0 & 1 & 1 & 0 & 0 & 0 
\end{array}\right],
 \end{gather}
respectively. We observe that $\Lset_1(3,3)$, has the same number of columns and unique columns as of $\Lset(2,3)$. The Loop $\Lset_2(3,3)$ has two columns and one unique column and $\Lset_3(3,3)$ has two columns and two unique columns more than $\Lset(2,3)$.
\end{Example}

Example~\ref{ex:const_p} shows that depending on the concatenated unimodal loop at Step~3 of Const.~\ref{con:child}, the difference between the number of columns (or unique columns) of the child and parent loop can vary. In the next Proposition, we quantify the maximum 
changes between a child loop and its parent loop in terms of the number of columns and unique columns and provide unimodal loops that can be used at the loop concatenation step (Step~3) that achieve this maximum changes.

\begin{Proposition}[\textbf{Max Addition Loop Concatenation}]
\label{prop:const}
For the purpose of constructing a max cardinality loop, $\MCLset$, we restrict our attention to the following two cases regarding the concatenated unimodal binary loop at Step~3 of Const.~\ref{con:child} (Thm.~\ref{thm:gcc}).

    \textbf{Case 1}: If all the columns in a sub-loop are identical, then, the concatenated loop can at most increase the number of columns and the number of unique columns in $\Lset(b,d)$ by two and one, respectively, compared to $\Lset(b-1,d)$. This is achievable using a unimodal binary loop whose first and last flip positions as defined in Def.~\ref{def:flfp} satisfy $i_f = i_l$ and ${\rm mod}(i_f,3) = 2$. 
    
    \textbf{Case 2}: If a sub-loop has two or more different columns, then, the concatenated loop can at most increase the number of columns and unique columns in $\Lset(b,d)$ each by two compared to $\Lset(b-1,d)$. This is achievable using a unimodal binary loop whose first and last flip positions correspond to different columns in the sub-loop and ${\rm mod}(i_f,3) = {\rm mod}(i_l,3) = 2$.
\end{Proposition}
\begin{proof}
Since the bits in a unimodal loop only flip at positions $i_f$ and $i_l$, all the added repeated columns at Step~2 of the Const.~\ref{con:child} 
that are not the repetition of the columns corresponding $i_f$ and $i_l$ will be omitted at Step~5 of the construction. Therefore, the number of columns of ${\Lset}(b,d)$ can at most increase by two. Furthermore, if $i_f$ and $i_l$ correspond to identical columns, the created new columns (if any) in $\Lset(b,d)$ will be the same. As a result, the number of unique columns in $\Lset(b,d)$ can increase at most by two compared to $\Lset(b-1,d)$ if $i_f$ and $i_l$ correspond to two different columns in the sub-loop and by one, otherwise.
\end{proof}

\subsection{Optimal Loop of Feedback Sequences}
\label{subsec:olfs}
As discussed, we are interested in the max cardinality unimodal loop for a given BA duration $b$ and delay $d$. So far, we have provided a construction for child loops from a parent loop (Const.~\ref{con:child} 
of Thm.~\ref{thm:gcc}
) and quantified the maximum number of columns and unique columns that can be added (per sub-loop) in this construction to the child loop (Prop.~\ref{prop:const}). 

To find a max cardinality loop for any $b$ and $d$, let us start from $\MCLset(1,d) = \odot[0,1]$ and increase duration $i$ by one at a time using Const.~\ref{con:child},
 until we reach $b$ while applying Prop.~\ref{prop:const}. We have two possibilities. 

\textbf{When $\mathbf{d= 1}$}, for any $i$, each sub-loop formed at Step~1 of the Const.~\ref{con:child} 
only consists of identical columns. This means that for $d=1$, we can only (and always) apply Case~$1$ of Prop.~\ref{prop:const} for each sub-loop and each $i$. Therefore, we can get $\MCLset(b,1)$ using this method.

\textbf{When $\mathbf{d\geq 2}$}, for any $i$, we want to ensure that a sub-loop formed at Step~1 of the Const.~\ref{con:child} always corresponds to Case~2 in Prop.~\ref{prop:const}. However, this might not always be the case so we impose additional constraints motivated by the following example.  
\begin{Example}
\label{ex:constmax}
Consider the loop $\Lset(3,3)$ in Ex.~\ref{ex:const_p}. Using Const.~\ref{con:child} and Prop.~\ref{prop:const}, one can show that this is an $\MCLset(3,3)$. Let us now use Const.~\ref{con:child} and Prop.~\ref{prop:const} to get an $\MCLset(4,3)$. First, we form the sub-loops with the same prefix of length one which are as in \eqref{eq:exc0}. Next, we repeat each column three times. We have
 \begin{align}
& \widehat{\Lset}_1(3,3) =\odot \left[\def\arraystretch{0.5}\begin{array}{ c c c c c c c c c}
1 & 1 & 1 & 1 & 1 & 1 & 1 & 1 & 1\\
1 & 1 & 1 & 0 & 0 & 0 & 0 & 0 & 0\\
0 & 0 & 0 & 0 & 0 & 0 & 1 & 1 & 1 
\end{array}\right] \\
 &\widehat{\Lset}_2(3,3) =\odot \left[\def\arraystretch{0.5}\begin{array}{ c c c c c c c c c}
0 & 0 & 0 & 0 & 0 & 0 & 0 & 0 & 0\\
0 & 0 & 0 & 1 & 1 & 1 & 1 & 1 & 1\\
1 & 1 & 1 & 1 & 1 & 1 & 0 & 0 & 0
\end{array}\right].
 \end{align}
 
Next, consider the following unimodal binary loops $\Pset_1 =$ $\odot[0$ $0$ $0$ $0$ $1$ $1$ $1$ $0$ $0]$ and $\Pset_2 =$ $ \odot[0$ $1$ $1$ $1$ $0$ $0$ $0$ $0$ $0]$, which satisfy the conditions in Prop.~\ref{prop:const} Case~2 and concatenate them to $\widehat{\Lset}_2(3,3)$ and $\widehat{\Lset}_2(3,3)$, respectively. We get 
 \begin{align}
     & \underline{\Lset}(4,3) =\odot \left[\def\arraystretch{0.5}\begin{array}{ c c c c c c c c c c c c c c c c c c}
1 & 1 & 1 & 1 & 1 & 1 & 1 & 1 & 1 & 0 & 0 & 0 & 0 & 0 & 0 & 0 & 0 & 0\\
1 & 1 & 1 & 0 & 0 & 0 & 0 & 0 & 0 & 0 & 0 & 0 & 1 & 1 & 1 & 1 & 1 & 1\\
0 & 0 & 0 & 0 & 0 & 0 & 1 & 1 & 1 & 1 & 1 & 1 & 1 & 1 & 1 & 0 & 0 & 0\\
0 & 0 & 0 & 0 & 1 & 1 & 1 & 0 & 0 & 0 & 1 & 1 & 1 & 0 & 0 & 0 & 0 & 0 
\end{array}\right] 
 \end{align}
 Finally, by removing the consecutive repetitions we get 
 \begin{align}
     & {\MCLset}(4,3) =\odot \left[\def\arraystretch{0.5}\begin{array}{ c c c c c c c c c c c c c c c c c c}
 1 & 1 & 1 & 1 & 1 & 0 & 0 & 0 & 0 & 0\\
 1 & 0 & 0 & 0 & 0 & 0 & 0 & 1 & 1 & 1\\
 0 & 0 & 0 & 1 & 1 & 1 & 1 & 1 & 1 & 0\\
 0 & 0 & 1 & 1 & 0 & 0 & 1 & 1 & 0 & 0 
\end{array}\right], 
 \end{align}
which is a max cardinality loop for $b=4$ and $d=3$. Now, let's increase $b$ further. Using Const.~\ref{con:child}, we will have four sub-loops, one of which is 
 \begin{align}
 \label{eq:osl}
     \tilde{\Lset}_1(4,3) =\odot \left[\def\arraystretch{0.5}\begin{array}{ c c c}
1 \\
1 \\
0 \\
0
\end{array}\right].
\end{align}
This sub-loop only consists of one column and so we cannot use Case~2 in Prop.~\ref{prop:const}. This happens since neither the first or last flip positions, $i_f$ or $i_l$ of the loop $\Pset_1$ fell on the column with the prefix $11$. As a result, the number of columns with the prefix $11$ remains one which leads to the one column sub-loop in \eqref{eq:osl}. To solve this issue, we can use the unimodal loop $\Pset_3 =$ $\odot\{0$ $1$ $1$ $1$ $0$ $0$ $0$ $0$ $0\}$ instead of  $\Pset_1$. Hence, we get 
\begin{align}
\label{eq:maxL2}
     & {\MCLset}(4,3) =\odot \left[\def\arraystretch{0.5}\begin{array}{ c c c c c c c c c c c c c c c c c c}
 1 & 1 & 1 & 1 & 1 & 0 & 0 & 0 & 0 & 0\\
 1 & 1 & 0 & 0 & 0 & 0 & 0 & 1 & 1 & 1\\
 0 & 0 & 0 & 0 & 1 & 1 & 1 & 1 & 1 & 0\\
 0 & 1 & 1 & 0 & 0 & 0 & 1 & 1 & 0 & 0 
\end{array}\right].
\end{align}

Now, we observe that all prefixes of length two have at least two unique columns, and if we want to further increase $b$, we can use Case~2 in Prop.~\ref{prop:const}. 
\end{Example}

Generalizing  this example, to make sure that each sub-loop includes at least two unique columns, at Step~3 of Const.~\ref{con:child}, we need to use a unimodal loop whose first and last flip positions $i_f$ and $i_l$ correspond to the columns (if any) with different prefixes of length $i-d+1$. On the other hand, the maximum number of unique prefixes of length $i-d+1$ in a sub-loop is two since each sub-loop consists of columns with the same prefix of length $i-d$ and the $(i-d+1)^{\rm th}$ bit is either one or zero. Therefore, we have the following result.

\begin{Theorem}[\textbf{Max Cardinality Loop}]
\label{thm:max_size}
The following construction results in an $\MCLset(b,d)$.
\begin{Construction}
\label{con:maxC}
Start from $i = 1$, set $\Lset(1,d)=$ $\odot\{1,0\}$. Increase the value of $i$ by one till we have $i =b$. For each $i$ use Const.~\ref{con:child} and generate the concatenated binary unimodal loop at Step~3 of Const.~\ref{con:child} as follows. 
\begin{itemize}
    \item For $d=1$, use a unimodal loop in which first and last flip positions satisfy $i_f = i_l$ with ${\rm mod}(i_f,3) = 2$.
    \item For $d\geq 2$, use a unimodal loop whose first and last flip positions fall on columns with different prefixes of length $i-d+1$ (if any, else any two different columns) with ${\rm mod}(i_f,3) = {\rm mod}(i_l,3) = 2$. 
\end{itemize}
\end{Construction}
\end{Theorem}
Next, we provide the exact number of columns and unique columns for max cardinality loops.

\begin{Theorem}
\label{thm:size}
Denoting the maximum number of columns and unique columns of a max cardinality $(b,d)$-unimodal loop by $N\str(b,d)$ and $M\str(b,d)$, respectively, we have 
\begin{itemize}
    \item For $d=1$,
\begin{align}
\label{eq:maxd1}
M\str(b,1) &= 2^b, \quad N\str(b,1) =  2^{b+1}-2.
\end{align} 
\item For $d\geq 2$,
\begin{align}
\label{eq:maxd}
\begin{aligned}
M\str(b,d)&= N\str(b,d) = \begin{cases}
M\str(b-1,d) + 2M\str(b-d,d) &b>d, \\
2b  &b \leq d.
\end{cases}
\end{aligned}
\end{align}
\end{itemize}
\end{Theorem}
\begin{proof}
The proof is provided in Appendix~\ref{app:size}.
\end{proof}

We conclude this section by proving the optimality of max cardinality loops.

\begin{Theorem}[\textbf{Optimal Unimodal Loop}]
\label{thm:opmax}
For a given $b$ and $d$, to find an optimal solution $\Sset\str(b,d)$ to the problem in \eqref{eq:optimization}, it is sufficient to use a {max cardinality} $(b,d)$-unimodal loop.
\end{Theorem}
\begin{proof}
The proof is provided in Appendix~\ref{app:opmax}.
\end{proof}

Theorem \ref{thm:opmax} indicates that an $\MCLset$, can be used as an optimality achieving loop of feedback sequences $\Lset\str$.

From Thm.~\ref{thm:size}, we observe that $M\str(b,d) = N\str(b,d)$ for $d\geq 2$. As a result, for $d\geq 2$ no $\MCLset$ includes any repetitions. In our BA problem, this means that the max cardinality loop resulting from Const.~\ref{con:maxC}, leads to contiguous URs. This is important in practice since contiguous beams are easier to implement compared to non-contiguous beams. We will provide a method for deriving these contiguous beams in the next section along with a lower-bound on the optimal performance.

\section{Bound on Maximum Expected Beamforming Gain and Achievability}
\label{sec:angular_Int}

In this section, we first provide an optimality achieving beam alignment procedure. Then, we provide a lower bound on the maximum expected beamforming gain.

\subsection{Achievability for Maximum Expected Beamforming Gain}

In Thm.~\ref{thm:opmax}, we have shown that to find an optimal scanning beams set $\Sset\str$, it is enough to consider a $(b,d)$-unimodal loop, $\Lset\str$, constructed using Const.~\ref{con:maxC}. Here, we find an optimal loop of component beams $\Iset\str$ which after combination with $\Lset\str$ through Const.~\ref{cons:Prob_beam_set} results to an optimal scanning beam set $\Sset\str$. 

To find the optimal loop of component beams, $\Iset\str = \{I_j\str\}_{j =1}^{N\str(b,d)}, I_j\str = (x_j\str , x_{j+1}\str]$, let us consider the set of URs $\mathcal{U}\str = \{U_k\str\}_{k=1}^{M\str(b,d)}$. The angular interval of $U_k\str$ is equal to the union of the component beams in $\Iset\str$ whose corresponding feedback is the $k^{\rm th}$ unique column of $\Lset\str$. We can find the optimal $x_j\str, j\in [N\str(b,d)]$ by rewriting the optimization in \eqref{eq:optimization} as follows. 
\begin{align}
\label{eq:OpI}
\begin{aligned}
&\argmin_{x_j, j\in [N\str(b,d)]} \quad \!\!\!\sum_{i=1}^{M\str(b,d)} \frac{1}{\sum_{j:\{I_j \in U_i\}} (x_{j+1} - x_j)} \!\!\int_{u_i}\!\!\! f_{\Psi}(x)dx. \\
&\text{such that:}~ \sum_{j:\{I_j \in U_i\}} (x_{j+1} - x_j) >0
\end{aligned}
\end{align}
For $d\geq 2$, since we know that each $U_k\str$ is contiguous and therefore consists of one component beam, the above optimization reduces to
\begin{align}
\label{eq:OpCon}
\begin{aligned}
&\argmin_{x_j, j\in [M\str(b,d)]} \quad  \sum_{j=1}^{M\str(b,d)} \frac{1}{x_{j+1} - x_j}\int_{x_j}^{x_{j+1}} f_{\Psi}(x)dx. \\
&\text{such that:}\quad\quad x_{j+1} - x_j >0
\end{aligned}
\end{align}
The construction of an optimal scanning beam set is summarized in the following result
\begin{Theorem}[\textbf{Optimal Scanning Beams}]
\label{thm:optBeamDesign}
An optimal set of scanning beams in \eqref{eq:optimization} can be generated using Const.~\ref{cons:Prob_beam_set} with a $(b,d)$-unimodal loop from Const.~\ref{con:maxC} and a loop of component beams derived form solving \eqref{eq:OpI} for $d=1$ or \eqref{eq:OpCon} for $d\geq 2$.
\end{Theorem}

\subsection{Maximum Expected Beamforming Gain}
Next theorem bounds the maximum expected beamforming gain for the BA problem.
\begin{Theorem}[\textbf{Maximum Expected Beamforming Gain}]
\label{thm:bounds}
The optimal value of the objective function in optimization problem \eqref{eq:optimization} when contiguous scanning beams are used is bounded as 
\begin{align}
\label{eq:con}
M\str(b,d) \leq \max_{\Sset} \mathbb{E}_{\Psi}\left[\frac{2\pi}{|\Beam(\Sset,\Psi)|}\right]
\end{align}
\end{Theorem}
\begin{proof}
To prove this results, we first provide an upper bound for minimum expected beamwidth $ \min_{\Sset} \mathbb{E}_{\Psi}|\Beam(\Sset,\Psi)| $ and then use the inequality 
\begin{align}
\label{eq:maxExpb}
    \max_{\Sset} \mathbb{E}_{\Psi}\left[\frac{2\pi}{|\Beam(\Sset,\Psi)|}\right] \geq \frac{2\pi}{\min_{\Sset} \mathbb{E}_{\Psi}\left[|\Beam(\Sset,\Psi)|\right]}
\end{align}
To upper bound $\min_{\Sset} \mathbb{E}_{\Psi}|\Beam(\Sset,\Psi)|$, we provide an achievability scheme as follows. We use Const.~\ref{con:maxC} to form a max cardinality loop $\MCLset$. Next, we form a loop of component beams $\Iset$ by partitioning $(0, 2\pi]$ into $M\str(b,d)$ equal length parts. Then, we use Const.~\ref{cons:Prob_beam_set} to create a set of scanning beams $\Sset$ based on the pair $(\MCLset, \Iset)$. It is easy to check that the length of each resulted UR from $\Sset$ is $\frac{2\pi}{M\str(b,d)}$. As a result,
\begin{align}
\label{eq:minExpw}
    \min_{\Sset} \mathbb{E}_{\Psi}\left[|\Beam(\Sset,\Psi)|\right] \leq \frac{2\pi}{M\str(b,d)}.
\end{align}
substituting \eqref{eq:minExpw} in \eqref{eq:maxExpb} completes the proof.
\end{proof}

By solving the optimizations in \eqref{eq:OpI} and \eqref{eq:OpCon} for the case of the uniform $f_{\Psi}(\cdot)$, we observe that for the case of uniform prior, the lower bound in \eqref{eq:con} is tight.

We conclude this section by noting that the derived results and analysis are not just limited to the expected beamforming gain objective. In fact, one can easily extend the proposed scanning beam set construction method to other objectives such as minimizing expected beamwidth which corresponds to solving the following optimization. 
\begin{align}
\begin{aligned}
\Sset^* = \argmin_{\Sset} \mathbb{E}_{\Psi}\left[|\Beam(\Sset,\Psi)|\right].
\end{aligned}
\end{align}
The only difference from our proposed construction would be that in \eqref{eq:OpI} and \eqref{eq:OpCon}, the terms $\dfrac{1}{\sum_{j:\{I_j \in U_i\}} (x_{j+1} - x_j)}$ and $\dfrac{1}{x_{j+1} - x_j}$ should be replaced by $ \sum_{j:\{I_j \in U_i\}} (x_{j+1} - x_j)$ and $ x_{j+1} - x_j$, respectively. We have also derived a lower bound for the minimum expected beamwidth in \cite{fullVersion}.

\begin{figure}[t]
    \centering
    \includegraphics[width = 0.35\linewidth]{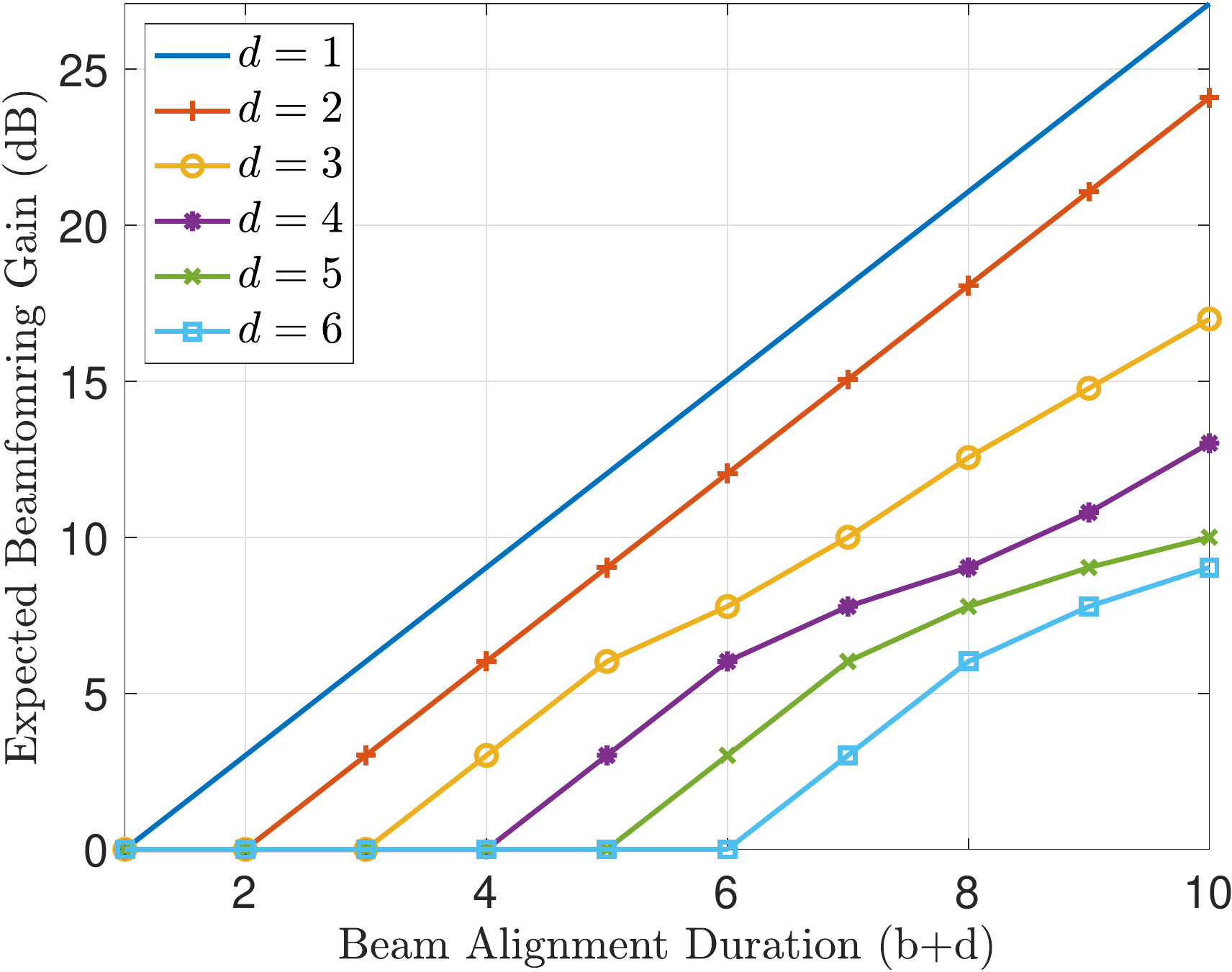}
    \caption{Maximum expected beamforming gain versus total duration of the beam alignment, illustrated for different delays. A uniform prior on $(0, 2\pi]$ is assumed. }
    \label{fig:beam_du1}
        \vspace*{-0.7cm}
\end{figure}

\section{Simulations and Numerical Analysis}
\label{sec:numerical}
We first investigate the trade-off between the achievable maximum expected beamforming gain, number of beam alignment packets $b$, and delay $d$. Let us consider uniform prior for the AoD $\Psi \sim \mathrm{Uniform} (0,2\pi]$. The maximum expected beamforming gain for different values of delay $d$ and total BA duration $b+d$ is illustrated in Fig.~\ref{fig:beam_du1}. As we can see, if the total duration of beam alignment is less than the delay, we wont receive any feedback and so the expected beamforming gain is one. We also observe that the cases of $d= 1$ and $d=2$ are parallel lines. The reason is that the case of $d=2$ leads to the same number of URs as of the case $d=1$. However, it has additional delay of one time-slot.  

Next, we compare the performance of the proposed BA method for the case of $\Psi \sim \mathrm{uniform}((0,$ $2\pi])$ with some of the state-of-the-art BA methods. We use a modified ES algorithm as described in the following to make sure our comparison is fair. More precisely, given $b$ probing packets, the ES method first divides $(0, 2\pi]$ into $b+1$ equal length contiguous URs. Then, transmits the $b$ beam alignment packets through the first $b$ URs. This way, it can distinguish all $b+1$ URs from one another based on the user's feedback and so achieves the expected beamwidth of $\frac{2\pi}{b+1}$. In the original ES method, however, $(0,2\pi]$ is divided into $b$ equal width regions and a BA packet is transmitted in all URs. Figure~\ref{fig:comp_naive}a shows the total duration of the BA for different target beamforming gains and different BA methods when $d=3$ and Fig.~\ref{fig:comp_naive}b shows the total BA duration that different BA methods require for different values of delay when the target beamforming gain is fixed. From Fig.~\ref{fig:comp_naive}a, we observe that the bisection method, which is optimal for $d=\{0,1\}$ is no longer optimal when we have feedback delay. Moreover, in some regimes, the non-interactive method derived in \cite{khalili2020optimal} which is optimal for $d \geq b$ and the modified ES method outperform the bisection method. This suggests that having small delay in the system can drastically lower the performance of interactive BA methods designed under the assumption of instant feedback. From Fig.~\ref{fig:comp_naive}b, we also have a similar observation that as delay increases, the performance of bisection rapidly degrades and after $d = 8$, it has the worst performance. Furthermore, this figure shows that for large values of delay, the optimal BA is the same as the optimal non-interactive BA. The reason is that the optimal non-interactive method in \cite{khalili2020optimal} is a special case of our problem for $d>b$.

\begin{figure}[t!]
    \centering
    \subfloat[]{\centering\includegraphics[width = 0.35\linewidth]{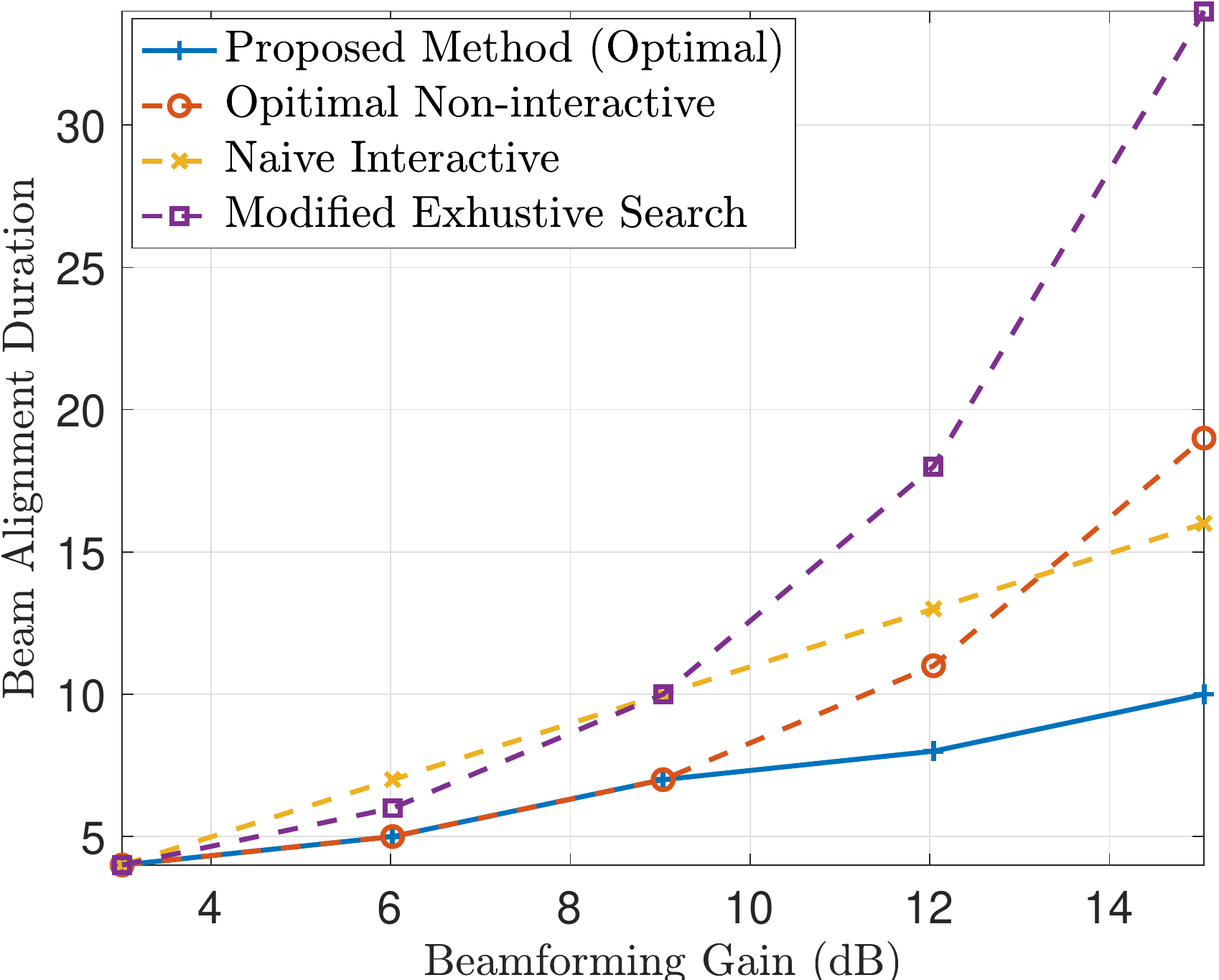}}
    \hspace{1.5cm}
   \subfloat[]{\centering\includegraphics[width = 0.35\linewidth]{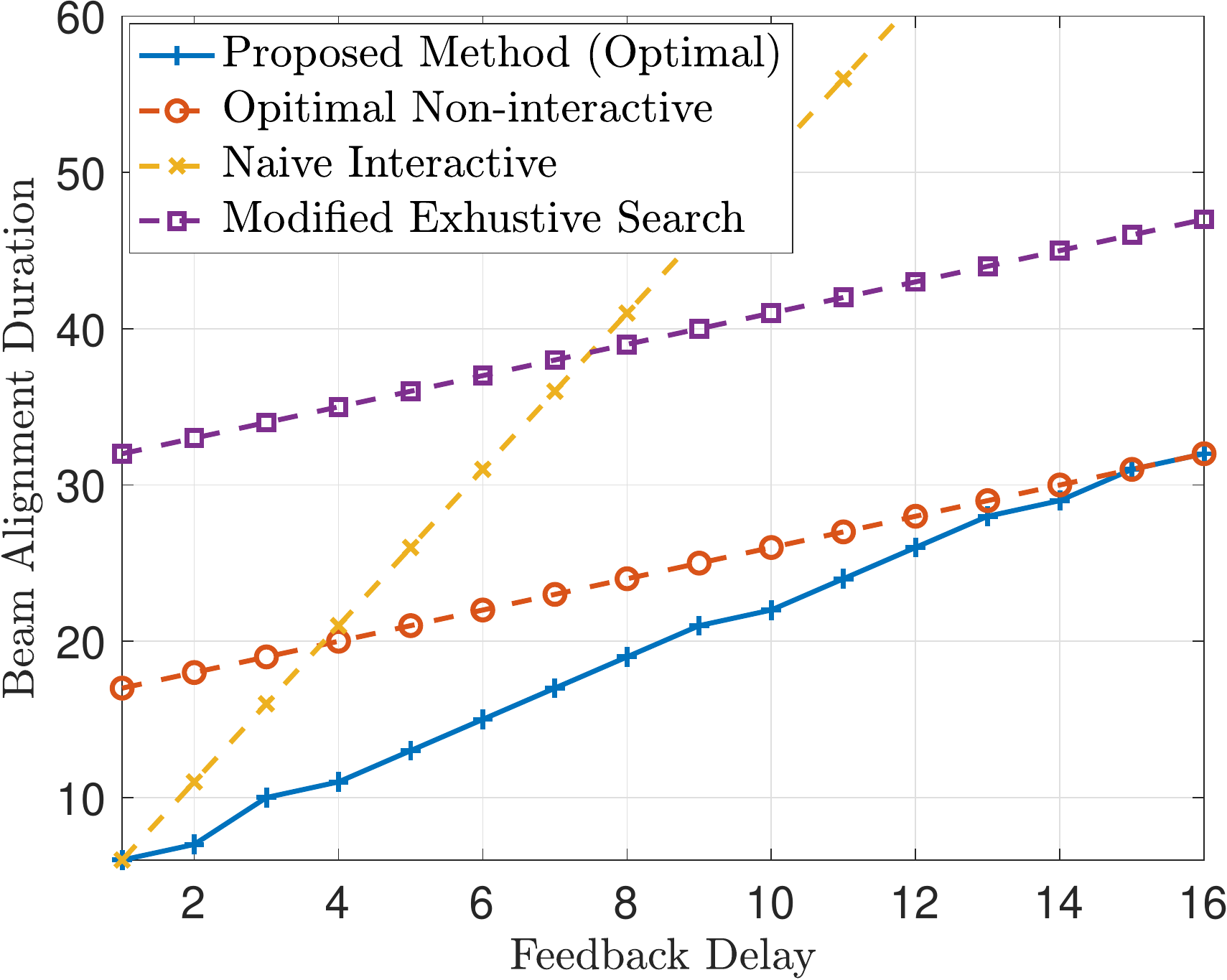}}
\caption{ (a) Total BA duration versus data beamwidth resolution for $d=3$. (b) Total BA duration versus feedback delays. The target beamforming gain is set to $32$ (beamwidth of $360/2^5 \approx 10^{\circ}$).}
\label{fig:comp_naive}
    \vspace*{-0.7cm}
\end{figure}

\begin{figure}[b!]
    \centering
    {\centering\includegraphics[width = 0.35\linewidth]{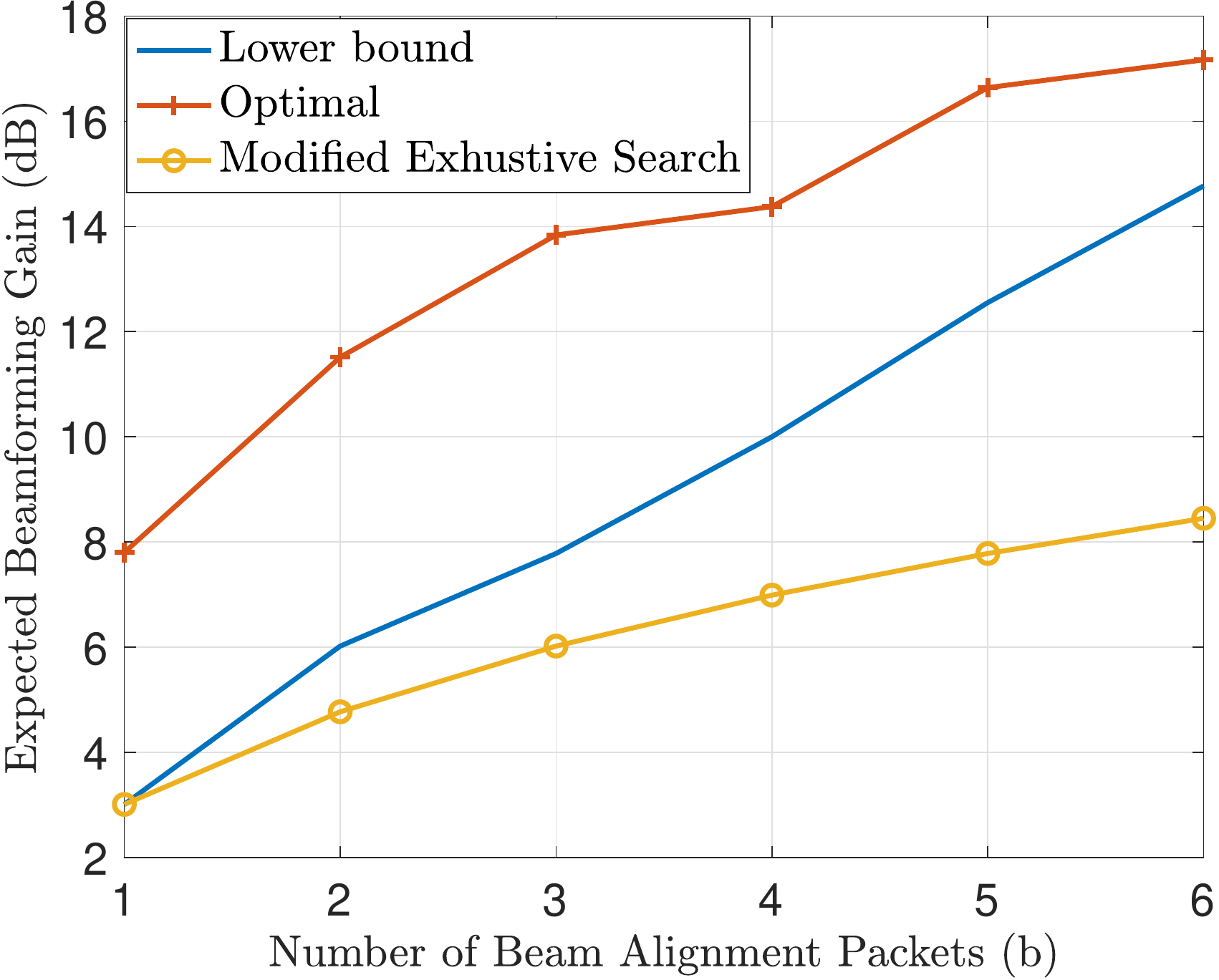}}
\caption{Performance of the optimal scanning beams Thm.~\ref{thm:optBeamDesign}, upper and lower bound in Thm.~\ref{thm:bounds}, and modified ES method for $d=3$.}
\label{fig:beam_comp}
\end{figure}

We conclude this section by investigating the performance of the optimal BA from Thm.~\ref{thm:optBeamDesign} with the lower bounds derived in Thm.~\ref{thm:bounds} and the performance of the modified ES method. For the comparison, we assume $d=3$ and consider a semi-Gaussian distribution for the AoD, where $\Psi \sim \frac{\mathcal{N}(\pi, 1)}{Q(-\pi)-Q(\pi)}$ for $\psi \in (0, 2\pi]$ with $Q(x) = \int_x^{\infty} \frac{1}{\sqrt{2\pi}}e^{-\frac{x^2}{2}}dx$. The comparison is plotted in Fig.~\ref{fig:beam_comp} where as expected we observe that the optimal performance is higher than the lower bound derived in Thm.~\ref{thm:bounds} and the performance of the modified ES method.
Comparing the formulas of the lower bound and modified ES beamforming gains, we observe that the optimal BA scheme is able to reduce the expected beamwidth at least $\frac{M\str(b,d)}{b} \geq 2$ times compared to the ES method.

\section{Conclusion}
In this paper, we have investigated the single-user analog beam alignment problem, where we have a fixed delay between each transmitted beam alignment packet and its received feedback given a fixed prior distribution on the AoD of the user. We have developed a general framework for this problem and provided a lower bound on the maximum expected beamforming gain. Moreover, we have developed explicit BA procedure that achieves the lower bound for the case of the uniform prior distribution. Furthermore, we have performed detailed simulations and numerical evaluations of the derived optimal BA strategy and compared its performance with the state-of-the-art methods. We have observed that the proposed BA method provides significant improvements in terms of BA duration required to achieve a fixed expected beamforming gain.

\bibliography{main.bbl}

\appendices

\section{Proof of Theorem \ref{thm:Lprop}}
\label{app:Lprop}

\textit{1)} Suppose that at time slot $i$, the BS receives a feedback sequence $(a_1,a_2,\ldots, a_{i-d})$ that translates to using the scanning beam $\Phi_i = S_{i,m}$. Consider the columns of $\Lset$ that have the prefix $a_1,a_2,\ldots, a_{i-d}$. The $i^{\rm th}$ bit of these columns would be one iff their corresponding component beams in $\Iset$ is included in the beam $S_{i,m}$. This is because, an ACK to the beam  $S_{i,m}$ means that the feedback $(a_1,a_2,\ldots, a_{i-d})$ was received by time-slot $i$ and the user AoD was included in the beam $S_{i,m}$. Therefore, row $i$ of the columns with prefix $a_1,a_2,\ldots, a_{i-d}$ is a sub-loop of the binary loop representing the beam $S_{i,m}$ based on the loop of the component beams. The binary loop representing the beam $S_{i,m}$ is unimodal and sub-loop of a unimodal loop is also unimodal.

\textit{2)} We show this by contradiction. Assume $\Lset$ has a consecutive repetition. This means two consecutive component beams lets say $I_1$ and $I_2$ have same feedback sequences. If so, all the scanning beams in $\Sset$ should either include both $I_1$ and $I_2$ or none.  Therefore, if we form the component beams of $\Sset$, since $I_1$ and $I_2$ are adjacent, we should get $I_1 \cup I_2$ as a component beam instead of two separate component beams $I_1$ and $I_2$ which is a contradiction.

\section{Proof of Theorem \ref{thm:1t1}}
\label{app:1t1}
We first show that the loop of component beams of the scanning beam set resulting from Const.~\ref{cons:Prob_beam_set}, $\Sset$, is the loop $\Iset$, and then prove that replacing the elements of the loop $\Iset$ with their corresponding feedback sequences from $\Sset$, results in the loop $\Lset$.

\textit{Loop of Component Beams:} To show that the loop of component beams of $\Sset$ is the loop $\Iset$, it suffices to show \textit{i)} no endpoints of the beams in $\Sset$ fall inside (not the edge) of the angular intervals in the loop $\Iset$  and \textit{ii)} for every two consecutive angular intervals in the loop $\Iset$, lets say $I_1$ and $I_2$, there is a beam in $\Sset$ that includes one but not the other. Item \textit{i)} holds by construction since for each angular interval in $\Iset$ a constructed beam in Step~3 of the Const.~\ref{cons:Prob_beam_set} either includes the angular interval completely or does not include it at all. For item \textit{ii)}, note that since $I_1$ and $I_2$ are adjacent, their corresponding columns in the loop $\Lset$ are different, otherwise the loop $\Lset$ would have consecutive repetitions which contradicts definition of $(b,d)$-unimodal loop (Thm.~\ref{thm:Lprop}). Let's denote these columns by $C_1$ and $C_2$. Assume that the first bit in which $C_1$ and $C_2$ differ is the $i^{\rm th}$ bit. Now,  consider Step~1 of Const.~\ref{cons:Prob_beam_set}. Since the prefixes of length $i-d$ of $C_1$ and $C_2$ are the same, these columns will be in the same sub-loop. However, since they differ in the $i^{\rm th}$ bit (row), the beam associated with their prefix of length $i-d$ only contains $I_1$ or $I_2$. 

\textit{Loop of Feedback Sequences:} 
We show that the feedback sequence based on the scanning beam set $\Sset$ to each of the component beams in the loop $\Iset$ is equal to the its corresponding column in the loop $\Lset$. Without loss of generality, consider a component beam $I_1$ from the loop $\Iset$ and let us denote its corresponding column in the loop $\Lset$ by $C_1$. Also, consider the $i^{\rm th}, i\in[b]$ bit of $C_1$ and suppose it is one (the case of it being zero can be argued similarly). While constructing $\Sset_i$ in Const.~\ref{cons:Prob_beam_set}, since bit $i$ of $C_1$ is one, the scanning beam in $\Sset_i$ corresponding to $C_1$ (the scanning beam which is designed for the feedback sequence equal to the first $i-d$ bits of $C_1$) will contain the component beam $I_1$. Therefore, if the user AoD falls in the angular interval $I_1$, the resulting feedback to this scanning beam should be one. Hence, bit $i$ of the feedback sequence corresponding to the component beam $I_1$ based on the scanning beam set $\Sset$ and bit $i$ of $C_1$ are equal for all $i\in[b]$.

\section{Proof of Theorem~\ref{thm:gcc}}
\label{app:gcc}
To prove this theorem, we assume that we are given a child loop $\Lset(b,d)$ for the parent loop $\Lset(b-1,d)$ and we will use Const.~\ref{con:child} to generate the loop $\Lset(b,d)$ from $\Lset(b-1,d)$. To this end, we provide the concatenated unimodal loops at Step~3 of Const.~\ref{con:child} for each sub-loop of the loop $\Lset(b-1,d)$ created at Step~1 of the construction.

Note that, for each sub-loop of the parent loop with certain prefix of length $b-d$, there is sub-loop of the child loop with the same prefix. Without loss of generality, let's consider the sub-loop ${\Lset}_1(b-1,d)$ from the parent loop and its corresponding sub-loop from the child loop ${\Lset}_1(b,d)$ whose columns have the same prefix of length $b-d$ as of the columns in the sub-loop ${\Lset}_1(b-1,d)$. Let us repeat each column in ${\Lset}_1(b-1,d)$ three times and form the loop $\hat{\Lset}_1(b-1,d)$. Also, let us repeat the columns of ${\Lset}_1(b,d)$ to form the loop $\hat{\Lset}_1(b,d)$ in which all columns have exactly three consecutive columns with the same prefix of length $b-1$ bits. 
This is possible since there cannot be more than three consecutive columns in ${\Lset}_1(b,d)$ that have the same prefix of length $b-1$. The reason is that the last element of these columns which is the $b^{\rm th}$ bit can either be zero or one and so having more than three columns with the same prefix of length $b-1$ either means that a column has consecutive repetition or the loop of the last row of $\Lset_1(b,d)$ is not unimodal. Both of these cases contradict the definition of $(b,d)$-unimodal loops (Thm.~\ref{thm:Lprop}).

All prefixes of length $b-1$ in ${\Lset}_1(b,d)$ also exist in ${\Lset}_1(b-1,d)$ due to the definition of the parent loop. Therefore, 
the loops $\hat{\Lset}(b,d)$ and $\hat{\Lset}(b-1,d)$ have the same cardinality and the columns at the same positions in these loops have the same prefix of length $b-1$. Let us denote the binary loop consisting of the last bits of the columns in the loop $\hat{\Lset}(b,d)$ as $\Pset$. This binary loop is unimodal since repetition of the consecutive bits in a unimodal loop leads to another unimodal loop and the last row of $\Lset_1(b,d)$ is unimodal by definition. It is easy to check that if we generate a binary loop $\Pset$ as above for each sub-loop of $\Lset(b-1,d)$ and use it for that sub-loop at Step~3 of Const.~\ref{con:child}, we will get the child loop ${\Lset}(b,d)$ at the end of the construction.

\section{Proof of Theorem \ref{thm:size}}
\label{app:size}

Note that the number of sub-loops of the loop $\Lset(b-1,d)$ created at Step~1 of Const.~\ref{con:maxC} is by definition of a parent loop equal to $M(b-d,d)$ (number of unique column in its parent loop of order $d-1$ and feedback sequences received by time-slit $b-d$). When $b\leq d$, we only have one sub-loop and so we have $M\str(b-d,d) =  M(b-d,d) = 1$ when $b\leq d$. We consider the cases of $d = 1$ and $d\geq 2$ separately.

For $d = 1$, for each of the created sub-loops, as argued in Sec.~\ref{subsec:olfs}, the concatenated unimodal loop can at most and always, increase the number of columns and unique columns by two and one, respectively. Therefore, 
\begin{align}
 M\str(b,1) = M\str(b-1,1) + M\str(b-1,1),\\
 N\str(b,1) = N\str(b-1,1) + 2M\str(b-1,1).
\end{align}
Also, note that for $b=1$, the max cardinality loop is $\MCLset(1,d)= \odot[0,1]$. Solving these equations lead to \eqref{eq:maxd1}.

For $d\geq 2$, for each of the created sub-loops, the concatenated unimodal loop, can at most and always, increase the number of columns and unique columns by two each as argued in Sec.~\ref{subsec:olfs}. Therefore, 
\begin{align}
    M\str(b,d) = M\str(b-1,d) + 2M\str(b-d,d),\\
    N\str(b,d) = N\str(b-1,d) + 2M\str(b-d,d).
\end{align} 
Given $\MCLset(1,d)= \odot[0,1]$ we have $N\str(1,d) = M(1,d) = 2$. Therefore, 
\begin{align}
    N\str(b,d)  = M\str(b,d) = M\str(b-1,d) + 2M\str(b-d,d).
\end{align}
This along with $M\str(b-d,d) = 1$ when $b\leq d$ leads to \eqref{eq:maxd}.

\section{Proof of Theorem \ref{thm:opmax}}
\label{app:opmax}

Consider, we have a scanning beam set $\Sset$ corresponding with the pair $(\Iset,\Lset)$. We consider the cases of $d=1$ and $d\geq 2$, separately. 

\textbf{When $\mathbf{d = 1}$}: We show that $\Sset$ can be constructed using an $\MCLset$ and Const.~\ref{cons:Prob_beam_set}. First, note that $\Lset(b,1)$ cannot have more than one repetition of each unique column based on the child loop construction, Const.~\ref{con:child}. The reason is that as discussed in Sec.~\ref{subsec:olfs}, each created sub-loop at Step~1 of this construction, only includes a unique column and the added unimodal loop at Step~3 of the construction can only increase the number of columns and unique columns by two and one, respectively. Therefore, if the parent loop does not have more than one repetition of a column, its child loops also cannot have more than one repetition of a column. On the other hand, none of the possible parent loops of order $b-1$ which are $\odot[0]$, $\odot[1]$, and $\odot[0,1]$ have any repetitions. Thus, no $(b,1)$-unimodal loop includes a column that is repeated more than once. 

Based on Thm.~\ref{thm:size}, a max cardinality loop $\MCLset(b,1)$ includes all the $2^b$ possible binary columns of length $b$ and has $2^{b+1}-2$ columns. In conjunction with the above discussion, we can conclude that in the MCL, there are $2$ columns which don't have a repetition and the rest of the columns are each repeated once. It is now evident that one can create an $\MCLset(b,1)$ from the loop $\Lset(b,1)$ by adding the binary columns of length $b$ which are not included in $\Lset(b,1)$ and a set of repetitions. Let us create a loop of component beams $\tilde{\Iset}$, such that an interval of length zero is added to the loop $\Iset$ in the same position of each column that we need to add to $\Lset(b,1)$ to get the $\MCLset(b,1)$. The scanning beams set $\tilde{\Sset}$ resulted form the pair ($\MCLset(b,1),\tilde{\Iset}$) using Const.~\ref{cons:Prob_beam_set} has the same URs as of the scanning beam set $\Sset$ and so has the same performance. 

\textbf{When $\mathbf{d \geq 2}$}: We show that given any scanning beam set $\Sset$, one can create a scanning beam set $\tilde{\Sset}$ using an $\MCLset$ and Const.~\ref{cons:Prob_beam_set} which has equal or higher expected beamforming gain. We create a loop of component beams $\tilde{\Iset}$ with the same number of elements as of $\MCLset(b,d)$ by setting the $|\Iset|$ elements of $\tilde{\Iset}$ equal to the intervals in $\Iset$ and the rest of the intervals to intervals of length zero. This is possible since $|\Lset| \leq |{\MCLset}|$ by Def.~\ref{def:MCL}. Then, using Const.~\ref{cons:Prob_beam_set} and pair $(\MCLset, \tilde\Iset)$ we create $\tilde\Sset$. The set of URs $\tilde\Uset$ corresponding to the $\tilde{\Sset}$ will have same contiguous angular intervals as of the $\Uset$. However, all the intervals would have a distinct feedback sequence since $\MCLset$ does not have any repetitions as a consequence of Thm.~\ref{thm:size} (i.e., $N\str = M\str$). Therefore, the expected beamforming gain corresponding to $\tilde\Sset$ is greater than or equal to that of $\Sset$.
\end{document}